\documentclass[sigconf,screen]{acmart}

\usepackage[capitalise]{cleveref}
\usepackage{hyperref}

\usepackage{tikz}

\usetikzlibrary{arrows}
\tikzset{
  big dot/.style={
    circle, inner sep=0pt, 
    minimum size=3pt, fill=purple
 }
}
\tikzset{
  bigbig dot/.style={
    circle, inner sep=0pt, 
    minimum size=5pt, fill=purple
 }
}
\tikzset{
  bigorange dot/.style={
    circle, inner sep=0pt, 
    minimum size=2.5pt, fill=orange
 }
}

\newcommand{\verteq}{\rotatebox{90}{$\,=$}}
\newcommand{\equalto}[2]{\underset{\scriptstyle\overset{\mkern4mu\verteq}{#2}}{#1}}

\usetikzlibrary{decorations.pathmorphing}

\usetikzlibrary{decorations.pathreplacing}

\tikzset{snake it/.style={decorate, decoration=snake}}

\usepackage{xspace}
\usepackage{todonotes}
\newcounter{todocounter}

\let\phi\varphi
\let\epsilon\varepsilon

\newcommand{\EM}{\mathop{\mathsf{EM}\vphantom{a}}\nolimits}
\newcommand{\Y}{\mathop{\mathsf{Y}\vphantom{a}}\nolimits}
\newcommand{\Last}{\mathop{\mathsf{L}\vphantom{a}}\nolimits}
\newcommand{\Yleq}[2]{\Y_{#1}\leq\Y_{#2}}
\newcommand{\Lleq}[2]{\Last_{#1}\leq\Last_{#2}}
\newcommand{\Llt}[2]{\Last_{#1}<\Last_{#2}}

\newcommand{\Levent}[2]{\Last_{#1}(#2)}
\newcommand{\Li}{\Last_{i}(t)}
\newcommand{\Lj}{\Last_{j}(t)}
\newcommand{\Lk}{\Last_{k}(t)}
\newcommand{\Pleq}[2]{\mathsf{Leq}({#1},{#2})}%
\newcommand{\Peq}[2]{\mathsf{Eq}({#1},{#2})}%
\newcommand{\LEQ}[4]{\mathsf{LEQ}_{{#1},{#2}}({#3},{#4})}
\newcommand{\EQ}[4]{\mathsf{EQ}_{{#1},{#2}}({#3},{#4})}
\newcommand{\p}[1]{\langle {#1} \rangle}
\newcommand{\test}[1]{{#1}?}
\newcommand{\lmove}{{\leftarrow}}

\newcommand{\pastpdl}{\mathsf{PastPDL}}
\newcommand{\locpastpdl}{\mathsf{LocPastPDL}}
\newcommand{\pastpdlconstants}{\mathsf{PastPDL[Y,L]}}
\newcommand{\locpastpdlconstants}{\mathsf{LocPastPDL[Y,L]}}

\newcommand{\True}{\top}
\newcommand{\False}{\bot}

\newcommand{\loc}{\mathsf{loc}}
\newcommand{\traces}[1][\Sigma]{\mathrm{Tr}(#1)}
\newcommand{\pset}{\mathscr{P}}

\newcommand{\da}{{\downarrow}}
\newcommand{\Da}{{\Downarrow}}

\newcommand{\A}{\mathcal{A}}
\newcommand{\B}{\mathcal{B}}

\newcommand{\s}{\Sigma}
\newcommand{\g}{\Gamma}

\newcommand{\alphabet}{\s}
\newcommand{\distribution}{\tilde{\Sigma}}

\newcommand{\ls}[1][S]{#1}

\newcommand{\lt}[1][\delta]{#1}
\newcommand{\gt}[1][\Delta]{#1}
\newcommand{\sinit}[1][s]{{#1}_{\mathrm{in}}}

\newcommand{\lab}{\Gamma}
\newcommand{\lalphabet}{\alphabet\times\lab}

\newcommand{\ltraces}{\traces[\lalphabet]}

\newcommand{\lc}{\circ_{\ell}}

\newcommand{\gossip}{\mathcal{G}}

\newcommand{\sdpath}{\textsf{sd}-path\xspace}

\newcommand{\on}[1]{\mathsf{on}_{#1}}
\newcommand{\prevon}[2]{\mathsf{prev}_{#1}(#2)}
\newcommand{\same}[2]{\mathsf{same}_{#1}(#2)}
\newcommand{\length}[1]{\|#1\|}
\newcommand{\Mod}[2]{\mathsf{Mod}({#1},{#2})}

\begin{document}

\copyrightyear{2024}
\acmYear{2024}
\setcopyright{acmlicensed}\acmConference[LICS '24]{39th Annual ACM/IEEE
Symposium on Logic in Computer Science}{July 8--11, 2024}{Tallinn, Estonia}
\acmBooktitle{39th Annual ACM/IEEE Symposium on Logic in Computer Science
(LICS '24), July 8--11, 2024, Tallinn, Estonia}
\acmDOI{10.1145/3661814.3662110}
\acmISBN{979-8-4007-0660-8/24/07}

\title{An expressively complete local past propositional dynamic logic over Mazurkiewicz traces and its applications}

\author{Bharat Adsul}
\email{adsul@cse.iitb.ac.in}
\orcid{https://orcid.org/0000-0002-0292-6670}
\affiliation{%
  \institution{IIT Bombay, India}
  \streetaddress{}
  \city{}
  \state{}
  \country{}
  \postcode{}
}

\author{Paul Gastin}
\email{paul.gastin@ens-paris-saclay.fr}
\orcid{https://orcid.org/0000-0002-1313-7722}
\affiliation{%
  \institution{Université Paris-Saclay, ENS Paris-Saclay, CNRS, LMF, 91190,
  Gif-sur-Yvette, France \and
CNRS, ReLaX, IRL 2000, Siruseri, India}
  \streetaddress{}
  \city{}
  \state{}
  \country{}
  \postcode{}
}

\author{Shantanu Kulkarni}
\email{shantanu3637@gmail.com}
\orcid{https://orcid.org/0009-0001-3525-2369}
\affiliation{%
  \institution{IIT Bombay, India}
  \streetaddress{}
  \city{}
  \state{}
  \country{}
  \postcode{}
}

\author{Pascal Weil}
\email{pascal.weil@cnrs.fr}
\orcid{https://orcid.org/0000-0003-2039-5460}
\affiliation{%
  \institution{CNRS, ReLaX, IRL 2000, Siruseri, India \and
  Univ. Bordeaux, LaBRI, CNRS UMR 5800, F-33400 Talence, France}
  \streetaddress{}
  \city{}
  \state{}
  \country{}
  \postcode{}
}

    

\begin{abstract}
We propose a local, past-oriented fragment of propositional dynamic logic to reason about concurrent  scenarios modelled as Mazurkiewicz traces, and prove it to be expressively complete with respect to regular trace languages. Because of locality, specifications in this logic are efficiently translated into asynchronous automata, in a way that reflects the structure of formulas. In particular, we obtain a new proof of Zielonka's fundamental theorem and we prove that any regular trace language can be implemented by a cascade product of localized asynchronous automata, which essentially operate on a single process. 

These results refine earlier results by Adsul et al. which involved a larger fragment of past propositional dynamic logic and used Mukund and Sohoni's gossip automaton. Our new results avoid using this automaton, or Zielonka's timestamping mechanism and, in particular, they show how to implement a gossip automaton as a cascade product.
\end{abstract}
\maketitle

\section{Introduction}
Mazurkiewicz traces \cite{Mazurkiewicz_1977} (just called traces in this paper) are a well-established model of concurrent behaviours. Each specifies a partial order on events of a scenario
in which a fixed set of processes, organized along a distributed architecture, interact with each other via
shared actions.  Traces have been extensively studied, with a particular attention to the
class of regular trace languages \cite{TraceBook-DR-1995,DBLP:reference/hfl/DiekertM97}.
The significance of that class arises from its expressivity, which is equivalent to
Monadic-Second-Order (MSO) logic definability \cite{tho90traces, EbingerM-TCS96,
GuaianaRS-TCS92}, and from the fact that regular trace languages can be implemented by
asynchronous automata (also known as Zielonka automata) \cite{Zielonka-RAIRO-TAI-87}.
Asynchronous automata yield implementations which fully exploit the distributed
architecture of the processes.

Our main result establishes that a local, past-oriented fragment of Propositional Dynamic
Logic (PDL) for traces, called $\locpastpdl$, is expressively complete with respect to
regular trace languages.  Let us first comment on the significance of this logic.

PDL was introduced in \cite{flpdl} to reason about arbitrary programs. 
LDL, a purely future-oriented version of PDL, was developed in \cite{ldl} to reason about
properties of finite words, and was shown to be expressively complete with respect to 
regular word languages.
A more recent work \cite{purepastltlldl} studied a purely past-oriented 
version of LDL %
and showed that it is expressively complete and admits a singly 
exponential translation into deterministic finite-state automata which
is an exponential improvement over LDL.
PDL has also been applied to message sequence charts (MSC) and 
related systems in \cite{BolligKuskeMeinecke-LMCS2010, Mennicke-LMCS13}. More recently, a star-free
version of PDL interpreted over MSC was shown to be expressively complete with respect to first-order
logic \cite{BolligFortinGastin-JCSS-21}. 
These prior works have successfully demonstrated that PDL 
is a suitable logical formalism for writing specifications in a variety
of contexts. On the one hand, it naturally extends linear temporal logics 
by permitting richer path formulas to express regular specification 
patterns easily. On 
the other hand, it supports efficient translations into automata-theoretic models,
which are central to the resolution of many verification or synthesis problems.

The logic $\locpastpdl$ admits three types of local and past-oriented
formulas: trace formulas, event formulas and path formulas.
A trace formula is a boolean combination of basic trace formulas of the form $\EM_{i}\phi$
asserting that the last event of process $i$ satisfies the event formula $\phi$.
It is local and past in the sense that its satisfaction 
depends only on the past information available to process $i$.

The \emph{event formulas} reason about the causal past of events.
They are boolean combinations of (a) atomic checks of 
letters (or actions), (b) formulas of the form $\p{\pi}$ claiming the existence of a 
path $\pi$ starting at the current event.
The \emph{path formula} $\pi$, necessarily \emph{localized at a process}, allows to 
march along the sequence of past events in which that process 
participates, checking for regular patterns interspersed with local 
tests of other event formulas. 

A typical event formula contains inside it several localized path formulas which in turn
contain other event formulas through embedded tests.  This natural hierarchical structure
of $\locpastpdl$ formulas renders them easier to design and to understand.

It turns out that our logic $\locpastpdl$ is essentially a fragment of 
$\locpastpdlconstants$ which was introduced in 
\cite{AdsulGastinSarkarWeil-CONCUR22} %
and shown to be expressively complete.
The logic $\locpastpdlconstants$ also supports {\em additional constant 
event/trace formulas} to compare the leading events for each process. 
The expressive completeness 
result for $\locpastpdlconstants$ crucially exploits the presence 
of these additional constant formulas.
Keeping track of the ordering information between leading process events is a fundamental
and very difficult problem in concurrency theory.  Its solution as an asynchronous automaton,
called a gossip automaton
\cite{Zielonka-RAIRO-TAI-87,MukundSohoni-DC97}, is a key
distributed time-stamping mechanism and a crucial %
ingredient in many asynchronous automata-theoretic constructions
\cite{Zielonka-RAIRO-TAI-87,DBLP:conf/icalp/KlarlundMS94}.

The central result of this work is that $\locpastpdl$ is itself expressively complete.
This is done by eliminating constant formulas in $\locpastpdlconstants$ using entirely new techniques.

Another important contribution is an efficient translation of
$\locpastpdl$ formulas into local cascade products of localized asynchronous automata 
(this, without any claim about efficiency, was already implicit in \cite{AdsulGastinSarkarWeil-CONCUR22}).
In asynchronous automata, each process runs a finite local-state device and these
devices synchronize with each other on shared actions.
As in \cite{MukundSohoni-DC97, AdsulGSW-CONCUR20, AdsulGSW-LMCS22}, 
one can use asynchronous automata to also locally compute relabelling
functions on input traces, similar in spirit to the sequential
letter-to-letter transducers on words.
\begin{figure}[h]
    \centering
    \begin{tikzpicture}%
        
        \draw (2.8,0) rectangle (4,1.2);
        \node at (3.4,0.6) {$A_1$};
        \draw[-stealth'] (2.3, 0.6) -- (2.8,0.6) node [midway, above] {$t$};

        \draw (5,0) rectangle (6.2,1.2);
        \node at (5.6,0.6) {$A_2$};
        \draw[-stealth'] (4, 0.6) -- (5,0.6) node [midway, above] {\small ${A_1}(t)$};

        \draw (7.7,0) rectangle (8.9, 1.2);
        \node at (8.3,0.6) {$A_3$};
        \draw[-stealth'] (6.2, 0.6) -- (7.7,0.6) node [midway, above] {\small ${A_2}({A_1}(t))$};
        \draw[-stealth'] (8.9, 0.6) -- (10.8,0.6) node [midway, above] {\small ${A_3}({A_2}({A_1}(t)))$};

        \end{tikzpicture}%
    \caption{Cascade Product}
    \label{cascadeproduct}
\end{figure}
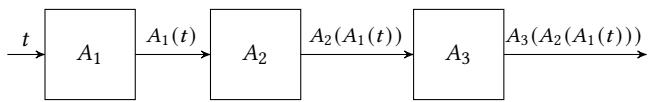

The {\em local cascade product} of asynchronous automata (or transducers)
from \cite{DBLP:conf/fsttcs/AdsulS04, AdsulGSW-CONCUR20, AdsulGSW-LMCS22} is a natural generalization
of cascade product in the sequential setting and, like in the word case 
\cite{StraubingBook}, it 
corresponds to compositions of related relabelling functions.
Note that, in a cascade product, the information flow is unidirectional
and hierarchical. 

The natural hierarchical structure of $\locpastpdl$ formulas makes it possible to
translate them in a modular way into a local cascade product of asynchronous automata.
The locality of $\locpastpdl$ is also an asset: each level in the cascade product
implementing a formula is a localized asynchronous automaton, that is, one where \emph{all
actions} take place on a \emph{single} process, and the other processes do nothing at
all.  In other words, these automata are essentially classical automata, operating on a
single process, and this is much simpler to understand and verify.  The past
orientation of $\locpastpdl$ makes our implementation deterministic.
Finally, a $\locpastpdl$ formula can be computed by a local cascade product of localized
asynchronous transducers/automata \emph{whose global-state space is singly exponential in
the size of the formula}.

The construction described above, coupled with the expressive completeness of
$\locpastpdl$ has some striking applications.  
It easily provides a new proof of Zielonka's theorem.
Another important application is a distributed
Krohn-Rhodes theorem. The classical 
Krohn-Rhodes theorem \cite{KR} implies that every regular word language
can be accepted by a {\em cascade product} of simple automata, namely, 
two-state reset automata and permutation automata. 
Our distributed version states that every regular trace language
can be accepted by a local cascade
product of localized two-state asynchronous reset automata and 
localized asynchronous permutation automata.

Additionally, we show that a gossip automaton/transducer can be implemented as a local cascade
product of localized asynchronous automata/transducers.  It is important to note that the
previously known implementations of gossip as an asynchronous transducer were
intrinsically non-local and non-hierarchical. They rely on a delicate
reuse of (boundedly many) time-stamps.  A process assigning the time-stamps
also needs to collect the information about which of its 
earlier time-stamps are in use by other processes. 
Given the `circular/self-referential' nature of this information flow, it 
appeared rather counter-intuitive that a gossip transducer could be 
implemented as a 
local cascade product (which is unidirectional in nature) of 
localized asynchronous automata (in which only one process is active).

We prove that $\locpastpdl$ is expressively complete by showing that the constant 
formulas in $\locpastpdlconstants$ can be defined in $\locpastpdl$ itself.
To do so, we develop an apparatus of special deterministic
path formulas, called \sdpath formulas, which suffice to address 
the leading events of processes in a bounded fashion. Then we
reduce the problem of checking the causal ordering of these leading events
to the problem of checking the causal ordering of the events addressed by \sdpath formulas.
We next show that the causal ordering of events addressed by \sdpath formulas
can be reduced to equality of such events, for possibly different and longer \sdpath
formulas.
Finally, the equality formulas are constructed with the help
of {\em separating} formulas.
Given a non-trivial \sdpath formula $\pi$, we  construct
a finite set $\Xi(\pi)$ of separating formulas which separate
every event $e$ from $\pi(e)$, the event referred from $e$
via address $\pi$.
The construction of these formulas, while intricate, is
entirely novel and constitutes the technical core of this work.

We now describe the organization of the paper.
In \cref{sec:locpastpdl}, we present the logics of interest, in
particular $\locpastpdl$, after
setting up basic preliminary notation.
\cref{sec:expcomplete} 
carries out the sophisticated task of expressing constant
formulas in $\locpastpdl$. The next \cref{sec:app} develops the efficient
translation into local cascade product of asynchronous devices and provides the
aforementioned applications of our expressive completeness result.
We finally conclude in \cref{sec:conclusion}.

\section{Local Past Propositional Dynamic Logic}\label{sec:locpastpdl}

\subsection{Basic notions about traces}
Let $\pset$ be a finite, non-empty set of processses. A distributed alphabet over $\pset$ is a family $\tilde{\Sigma}=\{\Sigma_{i}\}_{i \in \pset}$ of finite non-empty sets which may have non-empty intersection. The elements of $\Sigma_{i}$ are called the \emph{letters}, or \emph{actions}, \emph{that process $i$ participates in}. Let $\Sigma = \bigcup_{i\in \pset} \Sigma_{i}$.
Letters $a,b \in \alphabet$ are said to be \emph{dependent} if some process participates in both of them ($\exists i \in \pset, a, b \in \Sigma_{i}$); otherwise they are \emph{independent}. 

A poset is a pair $(E,\leq)$ where $E$ is a set and $\leq$ is a partial order on $E$. For $e,e' \in E$, $e'$ is an \emph{immediate successor} of $e$ (denoted by $e \lessdot e'$ ) if $e<e'$ and there is no $g \in E$ such that $e<g<e'$. A (\emph{Mazurkiewicz}) \emph{trace} $t$ over $\distribution$ is a triple $t= (E, \leq, \lambda)$, where $(E,\le)$ is
a finite poset, whose elements are referred to as \emph{events}, and $\lambda \colon E \to \Sigma$ is a labelling function which assigns a letter from $\alphabet$ to each event, such that
\begin{enumerate}
  \item for $e,e' \in E$, if $e \lessdot e'$, then $\lambda(e)$ and $\lambda(e')$ are dependent;
  \item for $e,e' \in E$, if $\lambda(e)$ and $\lambda(e')$ are dependent, then either $e \leq e'$ or $e' \leq e$.
\end{enumerate}

Let $t= (E,\leq,\lambda)$ be a trace over $\distribution$.  Let $i,j \in \pset$ be
processes.  Events in which process $i$ participates (that is: whose label is in
$\Sigma_i$), are called \emph{$i$-events} and the set of such $i$-events is denoted by
$E_{i}$.  $E_{i}$ is clearly totally ordered by $\leq$.  We define the \emph{location} of
an event $e \in E$ to be the set $\loc(e)$ of processes that participate in $e$, that is,
$\loc(e) = \{i \in \pset \mid \lambda(e) \in \Sigma_{i}\}$.  We slightly abuse notation to
define \emph{location} for letters as well: $\loc(a) = \{i \in \pset \mid a \in
\Sigma_{i}\}$.  The set of events in the past of $e$ is written $\da e = \{ f\in E \mid
f\leq e \}$ and the set of events in the strict past of $e$ is written $\Da e = \{ f\in E
\mid f<e \}$.  If $E_{i} \cap \Da e \neq \emptyset$, we denote by $\Y_{i}(e)$ the maximum
$i$-event in the strict past of $e$.  We denote by $\Y_{i,j}(e)$ the event
$\Y_{j}(\Y_{i}(e))$, if it exists.  If $E_{i} \neq \emptyset$, we denote by $\Li$ the
maximum (that is, last) $i$-event.  If in addition $E_{j} \cap \da \Li \neq \emptyset$, we
denote by $\Levent{i,j}{t}$ the maximum $j$-event in the past of the maximum $i$-event.

A set of traces over $\distribution$ is called a \emph{trace language}. Regular trace languages admit several characterizations, in particular in terms of MSO logic \cite{tho90traces} or of asynchronous (or Zielonka) automata \cite{Zielonka-RAIRO-TAI-87}.

\begin{example}
	\cref{fig:trace-diagram}%
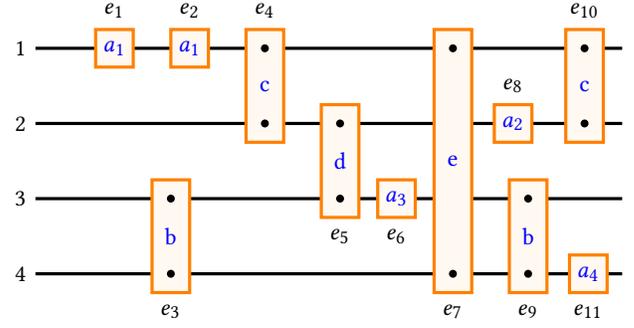
\begin{figure}[htb]
  \centering
  \begin{tikzpicture}
    \node at (0,4) {$1$};
    \node at (0,3) {$2$};
    \node at (0,2) {$3$};
    \node at (0,1) {$4$};
    \draw[black, very thick] (0.2,1) -- (8,1);
    \draw[black, very thick] (0.2,2) -- (8,2);
    \draw[black, very thick] (0.2,3) -- (8,3);
    \draw[black, very thick] (0.2,4) -- (8,4);

    \draw[orange, very thick, fill=orange!5] (1,3.75) rectangle (1.5,4.25);
    \node at (1.25,4.5) {$e_{1}$};
    \node at (1.25,4) {\textcolor{blue}{$a_{1}$}};

    \draw[orange, very thick, fill=orange!5] (2,3.75) rectangle (2.5,4.25);
    \node at (2.25,4.5) {$e_{2}$};
    \node at (2.25,4) {\textcolor{blue}{$a_{1}$}};

    \draw[orange, very thick, fill=orange!5] (1.75,0.75) rectangle (2.25,2.25);
    \node at (2,0.5) {$e_{3}$};
    \node at (2,1.5) {\textcolor{blue}{b}};
    \node[circle,fill=black,inner sep=0pt,minimum size=3pt] (e_{31}) at (2,2) {};
    \node[circle,fill=black,inner sep=0pt,minimum size=3pt] (e_{32}) at (2,1) {};

    \draw[orange, very thick, fill=orange!5] (3,2.75) rectangle (3.5,4.25);
    \node at (3.25,4.5) {$e_{4}$};
    \node at (3.25,3.5) {\textcolor{blue}{c}};
    \node[circle,fill=black,inner sep=0pt,minimum size=3pt] (e_{41}) at (3.25,4) {};
    \node[circle,fill=black,inner sep=0pt,minimum size=3pt] (e_{42}) at (3.25,3) {};

    \draw[orange, very thick, fill=orange!5] (4,1.75) rectangle (4.5,3.25);
    \node at (4.25,1.5) {$e_{5}$};
    \node at (4.25,2.5) {\textcolor{blue}{d}};
    \node[circle,fill=black,inner sep=0pt,minimum size=3pt] (e_{51}) at (4.25,3) {};
    \node[circle,fill=black,inner sep=0pt,minimum size=3pt] (e_{52}) at (4.25,2) {};

    \draw[orange, very thick, fill=orange!5] (4.75,1.75) rectangle (5.25,2.25);
    \node at (5,1.5) {$e_{6}$};
    \node at (5,2) {\textcolor{blue}{$a_{3}$}};

    \draw[orange, very thick, fill=orange!5] (5.5,0.75) rectangle (6,4.25);
    \node at (5.75,0.5) {$e_{7}$};
    \node at (5.75,2.5) {\textcolor{blue}{e}};
    \node[circle,fill=black,inner sep=0pt,minimum size=3pt] (e_{71}) at (5.75,1) {};
    \node[circle,fill=black,inner sep=0pt,minimum size=3pt] (e_{72}) at (5.75,4) {};

    \draw[orange, very thick, fill=orange!5] (6.3,2.75) rectangle (6.8,3.25);
    \node at (6.55,3.5) {$e_{8}$};
    \node at (6.55,3) {\textcolor{blue}{$a_{2}$}};

    \draw[orange, very thick, fill=orange!5] (6.5,0.75) rectangle (7,2.25);
    \node at (6.75,0.5) {$e_{9}$};
    \node at (6.75,1.5) {\textcolor{blue}{b}};
    \node[circle,fill=black,inner sep=0pt,minimum size=3pt] (e_{91}) at (6.75,2) {};
    \node[circle,fill=black,inner sep=0pt,minimum size=3pt] (e_{92}) at (6.75,1) {};

    \draw[orange, very thick, fill=orange!5] (7.25,2.75) rectangle (7.75,4.25);
    \node at (7.5,4.5) {$e_{10}$};
    \node at (7.5,3.5) {\textcolor{blue}{c}};
    \node[circle,fill=black,inner sep=0pt,minimum size=3pt] (e_{10_1}) at (7.5,4) {};
    \node[circle,fill=black,inner sep=0pt,minimum size=3pt] (e_{10_2}) at (7.5,3) {};

    \draw[orange, very thick, fill=orange!5] (7.3,0.75) rectangle (7.8,1.25);
    \node at (7.55,0.5) {$e_{11}$};
    \node at (7.55,1) {\textcolor{blue}{$a_{4}$}};

  \end{tikzpicture}
  \caption{Process trace diagram with labelled events.}
  \label{fig:trace-diagram}
\end{figure}
represents a trace with 11 events over 4 processes. The process set is $\pset = \{1,2,3,4\}$ and the event set is $E = \{e_1, \ldots, e_{11}\}$.
Each process is indicated by a horizontal line and time flows rightward. 
Each event is represented by a vertical box. Events are either local to a process, or they are non-local. Dots in the non-local events, indicate the participating processes. Each event is labelled by a letter in $\Sigma=\{a_{1},a_{2},a_{3},a_{4},b,c,d,e\}$. 
The distributed alphabet $\tilde{\Sigma}=\{\Sigma_{i}\}_{i \in \pset}$, is clear from the diagram, for example $\Sigma_{4}=\{b,e,a_{4}\}$. It is easy to infer causality relation ($\leq$) looking
at the diagram, for example, $e_{4}<e_{5}<e_{6}<e_{9}$.  We can also infer which
events are concurrent (not related by $\leq$): for example, event $e_{7}$ is concurrent to
events $e_{5}, e_{6}, e_{8}$, while the pairs of events $e_{5},e_{6}$ and
$e_{5},e_{8}$ are causally ordered by process $3$ and process $2$, respectively.

Taking a look at the maximum events of each process, we observe that $\Levent{1}{t}=e_{10}$,
$\Levent{2}{t}=e_{10}$, $\Levent{3}{t}=e_{9}$ and $\Levent{4}{t}=e_{11}$. The maximum event
for a process need not be a maximal event in the trace: for instance $\Levent{3}{t} = e_{9} < e_{11} = \Levent{4}{t}$. By definition, if $\Li$ is a $j$-event, then $\Levent{i,j}{t}=\Li$, as in the case of 
$\Levent{1,1}{t} = \Levent{1,2}{t}=\Levent{1}{t}=e_{10}$. If $\Li$ is not a $j$-event, then $\Levent{i,j}{t}= \Y_{j}(\Li) < \Li$: for instance $\Levent{1,3}{t}=\Y_{3}(\Levent{1}{t})
=e_{5}$ and $\Levent{1,4}{t}=\Y_{4}(\Levent{1}{t}) =e_{7}$.

Events of the form $\Y_{i}(e)$ and $\Y_{i,j}(e)$ are also directly visible, e.g.
$\Y_{4}(e_{6})=e_{3}$, $\Y_{3}(e_{8})=e_{5}$ and $\Y_{3,4}(e_{8}) = \Y_{4}(\Y_{3}(e_{8})) =\Y_{4}(e_{5})=e_{3}$. Note that $\Y_{i}(e)$ may not exist for every event,
for instance $\Y_{4}(e_{4})$ does not exist.
\end{example}

\subsection{Syntax and Semantics of $\locpastpdl$}

We first introduce the \emph{past propositional dynamic logic} $\pastpdl$ with the following syntax:
\begin{align*}
    \Phi & ::= \EM_{i}\varphi \mid \Phi\vee\Phi \mid \neg\Phi \\
	\varphi & ::= a \mid \varphi \lor \varphi \mid \neg \varphi \mid \p{\pi}\varphi \\ 
	\pi & ::= \lmove_i \,\mid \test{\varphi} \mid \pi + \pi \mid \pi\cdot\pi \mid \pi^* 
\end{align*}
\emph{Trace formulas} (of the form $\Phi$) are evaluated over traces, hence they define trace languages. \emph{Event formulas} (of the form $\varphi$) are evaluated at an event in a given trace and \emph{path formulas} (of the form $\pi$) are evaluated at a pair of events in a given trace. 

The semantics of $\pastpdl$ is as follows (ommitting the classical boolean connectives):
For a trace $t = (E, \leq, \lambda)$, events $e,f \in E$ and process $i \in \pset$, we let
\begin{align*}
  t &\models\EM_{i}\varphi &&\text{if $E_{i} \neq \emptyset$, and $t,\Li \models \varphi$,}
  \\
  t,e &\models a &&\text{if $\lambda(e)=a$,}
  \\
  t,e &\models\p{\pi}\varphi &&\text{if there exists an event $f$ such that} \\ &&&\text{$t,e,f \models \pi$ and $t,f \models \varphi$,}
  \\
  t,e,f &\models \lmove_i &&\text{if $e$ and $f$ are $i$-events and } \\ &&&\text{$e$ is the immediate successor of $f$ on process $i$,}
  \\
  t,e,f &\models \test{\varphi} &&\text{if $e=f$ and $t,e \models \varphi$,}
  \\
  t,e,f &\models \pi_1 + \pi_2 &&\text{if } t,e,f \models \pi_1 \text{ or } t,e,f
\models \pi_2
  \\
  t,e,f &\models \pi_1 \cdot \pi_2 &&\text{if there exists an event $g$ such that } \\ &&& t,e,g
	\models \pi_1 \text{ and } t,g,f \models \pi_2
  \\
  t,e,f &\models \pi^* &&\text{if there exist events } e=e_{0},e_{1},\ldots,e_{n}=f
  \\
  &&&\text{($n\ge 0$)} \text{ and } t,e_{i},e_{i+1}\models\pi \text{ for each } 0\leq i<n.
\end{align*}

We observe that path formulas $\pi$ can be seen as regular expressions over the (infinite)
alphabet consisting of the left moves $\lmove_i$ ($i\in\pset$) and the test formulas of
the form $\test{\varphi}$.  The left moves in this regular expression are called the
\emph{top-level moves} of $\pi$.  We say that $\pi$ is \emph{$i$-local} for some process
$i\in \pset$ if all its top-level moves are of the form $\lmove_{i}$.\footnote{%
Note that left moves $\lmove_j$ ($j \ne i$) may be used in the full description of an
$i$-local path formula $\pi$, if they occur in other path formulas which are used in event
formulas $\phi$ that are tested with $\test{\phi}$ in the regular expression defining
$\pi$.}
We then define $\locpastpdl$ to be the \emph{local} fragment of $\pastpdl$ where every
path formula is local --- that is, $i$-local for some $i$ (see \cite[Section
3]{AdsulGastinSarkarWeil-CONCUR22}).

\begin{example}
The trace in ~\cref{fig:trace-diagram} satisfies the following trace formula in $\locpastpdl$: 
  $$
  \underbrace{\EM_{1}(\p{\lmove_{1}\cdot\test{(d \lor 
  \p{\lmove_{4}}\True)}\cdot\lmove_{1}}c)}_{\phi_{1}} \lor 
  \underbrace{\EM_{3}(\p{\lmove_{4}}b)}_{\phi_{2}}.
  $$
  $\phi_{2}$ evaluates to $\False$ since if we move one event in the past on the process
  $4$ line from event $\Levent{3}{t} = e_{9}$, we encounter event $e_{7}$ which is not
  labelled $b$.  Now we examine $\phi_{1}$.  If we move one event in the past on the
  process $1$ line, from event $\Levent{1}{t} = e_{10}$, we are at event $e_{7}$.  At
  $e_{7}$ we check that backward movement is possible on process $4$ line, since $e_{3}$
  exists.  Hence the embedded test formula succeeds, and we move back again on process $1$
  line from event $e_{7}$, to see that we reach event $e_{4}$, which is labelled $c$.
  This shows that $\phi_{1}$ evaluates to $\True$ which in turn means the whole formula
  evaluates to $\True$.
\end{example}

The main result of this paper (Theorem~\ref{thm: locpastpdl is expressively complete}
below) is the expresssive completeness of $\locpastpdl$ with respect to regular trace
languages.  The proof uses the expressive completeness
\cite{AdsulGastinSarkarWeil-CONCUR22} of another variant of $\pastpdl$.  First we
consider the variant $\pastpdlconstants$, built on $\pastpdl$ using some additional
constants, defined by the following syntax and semantics.
\begin{align*}
	\Phi & ::= \EM\varphi \mid \Lleq{i}{j} \mid \Lleq{i,j}{k} \mid \Phi\vee\Phi \mid \neg\Phi \\
	\varphi & ::= a \mid \Yleq{i}{j} \mid \Yleq{i,j}{k} \mid \varphi \lor \varphi \mid \neg \varphi \mid \p{\pi} \\ 
	\pi & ::= \lmove_i \,\mid \test{\varphi} \mid \pi + \pi \mid \pi\cdot\pi \mid \pi^* 
\end{align*}
\begin{align*}
  t &\models\EM\varphi &&\text{if $t,e\models\varphi$ for some maximal event $e$ in $t$,}
  \\
  t &\models\Lleq{i}{j} &&\text{if $\Li,\Lj$ exist and 
  $\Li \leq \Lj$,} 
  \\
  t &\models\Lleq{i,j}{k} &&\text{if $\Levent{i,j}{t}, \Levent{k}{t}$ exist and $\Levent{i,j}{t} \leq \Lk$.}
  \\
	t,e &\models\Yleq{i}{j} &&\text{if $\Y_{i}(e),Y_{j}(e)$ exist and } \Y_{i}(e) \leq \Y_{j}(e)  \\
	t,e &\models\Yleq{i,j}{k} &&\text{if $\Y_{i,j}(e),\Y_{k}(e)$ exist and } \Y_{i,j}(e)\leq \Y_{k}(e) \\
	t,e &\models \p{\pi}    &&\text{ if there exists an event } f \text{ such that } t,e,f \models \pi \\
\end{align*}
Finally, we let $\locpastpdlconstants$ be the local fragment of $\pastpdlconstants$.
Notice that, apart from the modalities $\EM_{i}\varphi$
and $\p{\pi}\phi$, $\locpastpdl$ is a fragment of 
$\locpastpdlconstants$.

\begin{example}
  Consider again the trace $t$ from ~\cref{fig:trace-diagram}. It satisfies the following trace formulas in $\locpastpdlconstants$:
  \begin{itemize}
    \item $\Lleq{3}{4}$ since $e_{9} \leq e_{11}$
    \item $\lnot (\Lleq{2}{3})$ since $e_{10}$ is concurrent to $e_{9}$
    \item $\Lleq{2,3}{4}$ since $\Levent{2,3}{t} = \Y_{3}(\Levent{2}{t}) = \Y_{3}(e_{10}) = e_{5}$ which is below $\Levent{4}{t} = e_{11}$
    \item $\EM (\p{\lmove_{4}})$ since there is a maximal event $e_{11}$ from which a back move is possible on process $4$.
  \end{itemize}
  The trace $t$ does not satisfy the formula $\EM b$ as there is no maximal event labelled $b$ (the only maximal events in $t$ are $e_{10}$ 
  and $e_{11}$).
  The event $e_{8}$ in trace $t$, satisfies the following event formulas in  $\locpastpdlconstants$:
  \begin{itemize}
    \item $\Yleq{4}{3}$ since $\Y_{4}(e_{8}) = e_{3} < e_{5} = \Y_{3}(e_{8})$
    \item $\Yleq{2,3}{4}$ since $\Y_{2,3}(e_{8})= e_{3} = \Y_{4}(e_{8})$.
  \end{itemize}
\end{example}

\section{Expressive Completeness of $\locpastpdl$}\label{sec:technical-section}\label{sec:expcomplete}

The main result of this section is the following.

\begin{theorem}\label{thm: locpastpdl is expressively complete}
  A trace language is regular if and only if it can be defined by a $\locpastpdl$ sentence.
\end{theorem} 

The proof relies in part on the following statement \cite{AdsulGastinSarkarWeil-CONCUR22}.

\begin{theorem}\label{concur22 theorem}
  A trace language is regular if and only if it can be defined by a sentence in $\locpastpdlconstants$.
\end{theorem}

\begin{proof}[Proof of Theorem~\ref{thm: locpastpdl is expressively complete}]
  First, as in \cite[Proposition 1]{AdsulGastinSarkarWeil-CONCUR22}, it is easy to see 
  that for all $\locpastpdl$ sentence $\Phi$ we can construct an equivalent MSO sentence 
  $\overline{\Phi}$. Hence, $\locpastpdl$ sentences define regular trace languages.
  For the converse, by~\cref{concur22 theorem}, we only need to prove that the additional
  formulas in $\locpastpdlconstants$ are definable in $\locpastpdl$.  First observe that
  the modality $\EM\varphi$ is equivalent to
  $$\bigvee \limits_{i \in \pset} \left (\EM_i
  \varphi \land \neg \left(\bigvee \limits_{j\in\pset} \Llt{i}{j}\right) \right),$$
  where
  $\Llt{i}{j}=(\Lleq{i}{j})\wedge\neg(\Lleq{j}{i})$. Similarly, formula $\p{\pi}$ is equivalent to formula $\p{\pi}\True$. As a result, we only need to show that the constant formulas of $\locpastpdlconstants$ can be expressed in $\locpastpdl$, and this is the objective of the rest of this section. The constant event formulas $\Yleq{i}{j}$ and $\Yleq{i,j}{k}$ are dealt with in~\cref{Yconstants theorem}, and the constant trace formulas $\Lleq{i}{j}$ and $\Lleq{i,j}{k}$ are dealt with in~\cref{Lconstants theorem}.    
\end{proof}

Towards completing the proof and establishing~\cref{Yconstants theorem,Lconstants theorem}, we introduce the class of \emph{\sdpath formulas}. These are particular path formulas which uniquely identify certain past events, such as those of the form $\Y_i(e)$ (\cref{lem:Y_i and sdpath}). The constants $\Yleq{i}{j}$ and $\Yleq{i,j}{k}$ can then be expressed by formulas involving inequalities of \sdpath formulas (\cref{cor:comp-with-sdpaths1,cor:comp-with-sdpaths2}). These inequalities are then expressed by formulas involving equalities of \sdpath formulas (\cref{subsubsec:reduceIneqtoEq}), and finally by $\locpastpdl$-formulas (\cref{subsubsec:redEqtoSep}).

The last step uses what we call \emph{separating} $\locpastpdl$ event formulas. The construction of these formulas (\cref{subsubsec: construction separating}), while technical, is entirely novel and it replaces the use of timestamping (as in Zielonka's theorem \cite{Zielonka-RAIRO-TAI-87}) or the gossip automaton (in Mukund and Sohoni's work \cite{MukundSohoni-DC97}).

The constants $\Lleq{i}{j}$ and $\Lleq{i,j}{k}$ are expressed as $\locpastpdl$ trace
formulas in a similar fashion, see~\cref{subsec:eliminateTraceConstants}.
 
If $\pset$ is a singleton, $\pset=\{ \ell \}$, then $\Yleq{\ell}{\ell}$ is equivalent to $\p{\lmove_{\ell}}\True$ 
and $\Yleq{\ell,\ell}{\ell}$ is equivalent 
to $\p{\lmove_{\ell} \cdot \lmove_{\ell}}\True$. Further
both $\Lleq{\ell}{\ell}$ and $\Lleq{\ell, \ell}{\ell}$ are equivalent
to $\EM_{\ell} \True$. Therefore
$\locpastpdl$ and $\locpastpdlconstants$ are clearly equally expressive. In the remainder of this section, we assume that $|\pset| > 1$.

\subsection{Deterministic Path Formulas and their Properties}\label{subsec:detpathproperties}

Let us first set some notation and definitions.
A disjunction $\bigvee_{a\in\Gamma}a$ of letters from some $\Gamma\subseteq\Sigma$ is
called an \emph{atomic event formula}.
Note that Boolean combinations of atomic event formulas are (equivalent to) atomic event formulas. In particular, $\True$ and $\False$ are atomic event formulas:
$\True=\bigvee_{a\in\Sigma}a$ and $\False=\bigvee_{a\in\emptyset}a$.  Also, if $i\in
\pset$, then $\on{i}=\bigvee_{a\in\Sigma_{i}}a$ is an atomic event formula, which tests
whether an event is on process $i$.

If $\phi$ is an event formula, we let $\prevon{i}{\varphi}$ be the path formula $\lmove_i\cdot (\test{\neg\varphi}\cdot\lmove_{i})^{*}\cdot \test{\varphi}$. Then $t, e, f \models \prevon{i}{\varphi}$ if $e$ is an $i$-event and $f$ is the maximum $i$-event in the strict past of $e$ which satisfies $\phi$. For instance, $\prevon{i}{\on{i}}$ is equivalent to $\lmove_{i}$.

Finally, we say that a path formula $\pi$ is \emph{deterministic} if, for each trace $t$ and each event $e$
in $t$, there exists at most one event $f$ in $t$ such that $t,e,f\models\pi$.  We write
$\pi(e)=f$ when such an event $f$ exists. It is easily verified that any path formula of the form $\test\phi$, $\lmove_i$ or $\prevon{i}{\varphi}$ is deterministic, and that, if $\pi$ and $\pi'$ are deterministic, then so is $\pi\cdot\pi'$.

We now introduce the class of \emph{simple deterministic} path formulas 
(\sdpath formulas), whose syntax is the following:
$$\pi ::= \test{\varphi} \mid \prevon{i}{\varphi} \mid \pi\cdot\pi,$$
where $\varphi$ is an atomic event formula and $i \in \pset$.
By definition, an \sdpath formula $\pi$ can be seen as a word on the (finite) alphabet consisting of
tests $\test\phi$ and local path formulas $\prevon{i}{\phi}$, where $\phi$ is atomic. We then define the
\emph{length} of $\pi$ to be the number $\length{\pi}$ of ``letters'' of the form
$\prevon{i}{\phi}$ in $\pi$.  More formally, we let $\length{\test{\varphi}}=0$,
$\length{\prevon{i}{\phi}}=1$ and $\length{\pi\cdot\pi'}=\length{\pi}+\length{\pi'}$.

We record several important properties of \sdpath formulas in the following lemma.

\begin{lemma} \label{lem:properties}
Let $\pi$ be an \sdpath formula. Then
\begin{enumerate}
  \item $\pi$ is deterministic.

  \item $\pi$ is \emph{monotone}: %
  for each trace $t$ and for all events $e,e'$ in $t$, if $e\leq e'$ and both $\pi(e)$ and
  $\pi(e')$ exist, then $\pi(e)\leq\pi(e')$.

  \item The event formula $\p{\pi}$ is equivalent to a $\locpastpdl$ formula.
\end{enumerate}
\end{lemma}

\begin{proof}
The first two statements are easily proved by structural induction on $\pi$. The last statement is proved using the facts that $\p{\pi}$ is equivalent to $\p{\pi}\test\True$, that the \sdpath formulas $\test{\varphi}$ and $\prevon{i}{\varphi}$ are local, and that $\p{\pi\cdot\pi'}=\p{\pi}\p{\pi'}=\p{\pi\cdot\test{\p{\pi'}}}$.
\end{proof}

\begin{example}
  Note that \sdpath formulas used to address past events are not necessarily unique. Referring to ~\cref{fig:trace-diagram}, it is easy to verify the following examples.
  \begin{align*}
    \Y_{1}(e_{8}) &= (\prevon{2}{\True}\cdot\prevon{2}{\True}\cdot\test{\on{1}}) (e_{8})=e_{4} \\
    \Y_{1}(e_{8}) &= (\prevon{2}{\on{1}}) (e_{8}) = e_{4}\\
    \Y_{2}(e_{8}) &= (\prevon{2}{\True}\cdot\test{\on{2}}) (e_{8})= e_{5}\\
    e_{1} &= (\prevon{1}{\True}\cdot\prevon{1}{\True}) (e_{4})\\
    \Y_{3}(e_{8})&= (\prevon{2}{\True}\cdot\test{\on{3}})(e_{8}) =e_{5}\\
    \Y_{4}(e_8) &= (\prevon{2}{\on{3}}\cdot\prevon{3}{\on{4}})(e_{8}) =e_{3}
  \end{align*}
\end{example}

We say that a \sdpath formula $\pi$ is in \emph{standardized form} if it is of the form
$$\test{\phi'_{0}}\cdot\prevon{i_{1}}{\phi_{1}}\cdot\test{\phi'_{1}}
\cdot\cdots\cdot\test{\phi'_{n-1}}\cdot\prevon{i_{n}}{\phi_{n}}\cdot\test{\phi'_{n}}.$$
Since the path formulas $\prevon{i_1}{\phi_{1}} \cdot \prevon{i_2}{\phi_{2}}$ and
$\test{\phi}\cdot\test{\psi}$ are equivalent to $\prevon{i_1}{\phi_{1}} \cdot \test{\top}
\cdot \prevon{i_2}{\phi_{2}}$ and $\test{(\phi \land\psi)}$, respectively, we see 
that every \sdpath formula is equivalent to one in standardized form, of the same
length.
The following easy lemma will be useful in the sequel.

\begin{lemma}
For any $n \geq 0$, there are only finitely many logically distinct 
\sdpath formulas of length $n$.
\end{lemma}

\begin{proof}
Let $\pi$ be an \sdpath formula of length $n$. We may assume that 
$\pi$ is in standardized form:
$$\pi = \test{\phi'_{0}}\cdot\prevon{i_{1}}{\phi_{1}}\cdot\test{\phi'_{1}}
\cdot\cdots\cdot\test{\phi'_{n-1}}\cdot\prevon{i_{n}}{\phi_{n}}\cdot\test{\phi'_{n}}.$$
The $i_j$ range over the finite set $\pset$.  And the $\phi_i$ and $\phi'_i$ are atomic
event formulas, and hence range over a finite set (the power set of $\Sigma$).  The result
follows directly.
\end{proof}

The link between the constants $\Yleq{i}{j}$, $\Yleq{i,j}{k}$ and the \sdpath formulas is 
given by \cref{lem:Y_i and sdpath}. 
 
\begin{proposition}\label{lem:Y_i and sdpath}
Let $t$ be a trace $t$, $e$ an event and $i$ a process such that $\Y_{i}(e)$ exists.  Then
$\Y_{i}(e)=(\pi\cdot\test{\on{i}})(e)$ for some \sdpath formula $\pi$ satisfying $1 \le
\length\pi < |\pset|$.  %
\end{proposition}

\begin{proof}
Suppose that $f= \Y_{i}(e)$ exists. By definition of $\Y_{i}$, we have $f < e$. We construct 
sequences of $m$ events $f=e_1 < e_2 < e_3 < \cdots < e_m < e$ and $m$ {\em distinct} 
processes $i=i_1, i_2, i_3 \ldots, i_m$ 
such that
\begin{enumerate}
	\item For $1 \leq j \leq m$, $\Y_{i_{j}}(e) = e_j$. 
\item For $1 < j \leq m$, $e_{j-{1}}$ and $e_{j}$ are $i_j$-events, and
	$\loc(e_{j-{1}})\cap\loc(e)=\emptyset$.
\item There is a process which participates in both $e_m$ and $e$, that is, $\loc(e_m) \cap \loc(e) \neq \emptyset$.
\end{enumerate}
We begin the construction by letting $e_1 = f$ and $i_1 = i$.  In particular, we have
$Y_{i_1}(e) = e_1$.

Suppose that we have constructed a length $n$ sequence of events
$e_1 < e_2 < e_3 < \cdots < e_n < e$ and a length $n$ sequence of
distinct processes
$i_1, i_2, i_3 \ldots, i_n$ such that
(a) for $1 \leq j \leq n$, $\Y_{i_{j}}(e) = e_j$ and
(b) for $1 < j \leq n$, $e_{j-{1}}$ and $e_{j}$ are $i_j$-events, and
$\loc(e_{j-{1}})\cap\loc(e)=\emptyset$.
If $\loc(e_n) \cap \loc(e) \ne \emptyset$, we let $m = n$ and we are done.

If instead $\loc(e_n) \cap \loc(e) = \emptyset$, we extend these sequences as follows. 
We have $i_n \not\in \loc(e)$, since $i_n \in \loc(e_n)$.
Since $e_n < e$, there exists $e'$ such that $e_n \lessdot e' \leq e$.
As an immediate successor of $e_n$, $e'$ satisfies $\loc(e_n) \cap \loc(e') \neq \emptyset$ and, as a result, we have $e' < e$. Choose
$i_{n+{1}}$ to be any process in $\loc(e_n) \cap \loc(e')$.
As $e'$ is an $i_{n+{1}}$-event strictly below $e$, $\Y_{i_{n+{1}}}(e)$
exists. We now set $e_{n+{1}}=\Y_{i_{n+{1}}}(e)$. Clearly, both
$e_n$ and $e_{n+{1}}$ are $i_{n+{1}}$-events and $e_n < e' \leq e_{n+{1}} < e$. 
We now argue that,  for any $1 \leq j \leq n$, $i_{n+{1}} \neq i_{j}$; indeed, if $i_{n+1} = i_j$, then $e_{n+{1}}= \Y_{i_{n+{1}}}(e) = \Y_{i_{j}}(e)=e_j \leq e_n$, a contradiction.
Thus the processes $i_1, i_2, \ldots, i_n, i_{n+{1}}$ are pairwise distinct.

We repeat this procedure as long as $\loc(e)$ is disjoint from the location of the last event constructed. Since $\pset$ is finite, the procedure can be repeated only finitely many times. Let $m$ be the length of
the final event sequence $e_1 < e_2 < \cdots < e_m$ and the final process
sequence $i_1, i_2, \ldots, i_m$. At that stage, we have $\loc(e_m) \cap \loc(e) \neq \emptyset$. 

We now use these sequences of events and processes to construct the announced \sdpath formula $\pi$. 

The fact that $\Y_{i_m}(e)=e_m$ implies that no event $e'$ satisfying $e_m < e' < e$ is an $i_m$-event. Let $\ell \in \loc(e_m) \cap \loc(e)$. Then both $e_m$ and $e$ are $\ell$-events and $e_m = \Y_{i_m}(e) = \prevon{\ell}{\on{i_m}}(e)$.

Similarly, for $1 < j \leq m$, $\Y_{i_{j-1}}(e) = e_{j-1}$ and hence no $i_j$-event $e'$ satisfying  $e_{j-1} < e' < e$ is an $i_{j-1}$-event. In particular, $\Y_{i_{j-1}}(e_j) = e_{j-1}$. Moreover, both $e_{j-1}$ and $e_j$ are $i_j$-events, so $e_{j-1} = \Y_{i_{j-1}}(e_j) = \prevon{i_j}{\on{i_{j-1}}}(e_j)$.

Now let $\pi = \prevon{\ell}{\on{i_m}} \cdot 
\prevon{i_m}{\on{i_{m-1}}} 
\cdots  \prevon{i_2}{\on{i_1}}$. Then $\pi$ is an \sdpath formula and $Y_i(e) = f = e_1 = \pi(e)$. Moreover, since $i_1 = i$, we also have $Y_i(e) = (\pi\cdot\test{\on i})(e)$.

To conclude, we only need to verify that $m < |\pset|$. Let indeed $\ell$ be any process in $\loc(e_m) \cap \loc(e)$. Since $\loc(e_j) \cap \loc(e) = \emptyset$ for every $j < m$, we find that $\ell\notin\loc(e_j)$, and hence $\ell$ is distinct from $i_1, i_2, \ldots, i_m$. Therefore $m+1 \le |\pset|$, which completes the proof.
\end{proof}

\begin{example}
The proof of \cref{lem:Y_i and sdpath} shows that, for the trace in \cref{fig:trace-diagram}, we have
  $$\Y_{2}(e_{11}) = (\prevon{4}{\on{3}}\cdot\prevon{3}{\on{2}}\cdot\test{\on{2}}) (e_{11}) = e_{5}$$ 
\end{example}

\subsection{Expressing Constant Event Formulas in $\locpastpdl$}\label{subsec:eliminateEventConstants}

The main theorem in this section is the following.

\begin{theorem}\label{Yconstants theorem}
Let $i,j,k\in\pset$. The constant event formulas $\Yleq{i}{j}$ and $\Yleq{i,j}{k}$ can be expressed in $\locpastpdl$
\end{theorem}

\begin{proof}[Overview of the proof]
The proof relies on a complex construction, which occupies the rest of \cref{subsec:eliminateEventConstants}. We show (\cref{prop: existence of Pleq} below) that, for each pair $(\pi,\pi')$ of \sdpath formulas, there exists a $\locpastpdl$ event formula $\Pleq{\pi}{\pi'}$ which expresses the following: $t,e \models \Pleq{\pi}{\pi'}$ if and only if $\pi(e)$ and $\pi'(e)$ exist and $\pi(e) \le \pi'(e)$. \cref{lem:Y_i and sdpath} is then used to show that $\Yleq{i}{j}$ and $\Yleq{i,j}{k}$ are logically equivalent to $\locpastpdl$ event formulas using formulas of the form $\Pleq{\pi}{\pi'}$, see~\cref{cor:comp-with-sdpaths1,cor:comp-with-sdpaths2}.

The proof of \cref{prop: existence of Pleq} (in \cref{subsubsec:reduceIneqtoEq}) uses the existence of another class of $\locpastpdl$ formulas, written $\Peq{\pi}{\pi'}$, which express that $\pi(e)$ and $\pi'(e)$ exist and $\pi(e) = \pi'(e)$ (\cref{prop: existence of Peq} below). And the proof of \cref{prop: existence of Peq}, in \cref{subsubsec:redEqtoSep}, in turn uses the existence of finite sets of so-called separating formulas, which are constructed in~\cref{subsubsec: construction separating}.
\end{proof}

We now unravel this complex proof structure.

\begin{proposition}\label{prop: existence of Pleq}
For each pair of \sdpath formulas $\pi$ and $\pi'$,
there exists a $\locpastpdl$ event formula $\Pleq{\pi}{\pi'}$ such that $t,e \models
\Pleq{\pi}{\pi'}$ if and only if $\pi(e)$ and $\pi'(e)$ exist and $\pi(e) \le \pi'(e)$.
\end{proposition}

The proof of \cref{prop: existence of Pleq} is deferred to \cref{subsubsec:reduceIneqtoEq}. \cref{prop: existence of Pleq} yields the two following corollaries, which establish~\cref{Yconstants theorem}.

\begin{corollary}\label{cor:comp-with-sdpaths1}
Let $i,j \in \pset$ be processes. Then $\Yleq{i}{j}$ is logically equivalent to the following $\locpastpdl$ event formula:
\begin{align*}
  \underbrace{\bigvee_{\pi'} \p{\pi'\cdot\test{\on{i}}}}_{\varphi_{1}} 
  &\land \underbrace{\bigvee_{\pi} \p{\pi\cdot\test{\on{j}}}}_{\varphi_{2}} 
  \\ &\land \underbrace{\bigvee_{\pi} 
\bigwedge_{\pi'} \Big(\p{\pi'\cdot\test{\on{i}}} 
\implies \Pleq{\pi'\cdot\test{\on{i}}}{\pi\cdot\test{\on{j}}}\Big)}_{\varphi_{3}}
\end{align*}
where the disjunctions and conjunctions run over \sdpath formulas $\pi,\pi'$  of length at least $1$ and at most $|\pset|-1$
\end{corollary}

\begin{proof}
By definition, $t,e \models \Yleq{i}{j}$ if and only if $\Y_{i}(e)$ and $\Y_{j}(e)$ exist, and $\Y_{i}(e) \leq \Y_{j}(e)$.
If $\pi'$ is an \sdpath formula and $\pi'\cdot\test{\on{i}}(e)$ exists, then there is an
$i$-event in the strict past of $e$, and hence $\Y_{i}(e)$ exists.  Conversely,
\cref{lem:Y_i and sdpath} shows that if $\Y_{i}(e)$ exists, then it is equal to
$\pi'\cdot\test{\on{i}}(e)$ for some such $\pi'$, with length between 1 and $|\pset|-1$.
Therefore $t,e\models \phi_1$ if and only if $\Y_{i}(e)$ exists. Similarly $\phi_2$ expresses the existence of $\Y_{j}(e)$.

Now assume that $t,e \models \Yleq{i}{j}$.  Let $\pi$ be given by \cref{lem:Y_i and
sdpath} such that $\Y_{j}(e)= \pi\cdot\test{\on{j}}(e)$, and let $\pi'$ be an arbitrary
\sdpath with $1 \le \length{\pi'}< |\pset|$.
Consider the $(\pi,\pi')$ implication in $\phi_3$.  If $\pi'\cdot\test{\on{i}}(e)$ does not
exist, this implication is vacously satisfied.  If instead $\pi'\cdot\test{\on{i}}(e)$ exists,
then it is an $i$-event, so $\pi'\cdot\test{\on{i}}(e) \leq \Y_{i}(e)$ and the same
implication is also satisfied.
This establishes the fact that $t,e\models  \Yleq{i}{j}$ implies that $t,e\models \phi_3$.

Conversely, assume that $\Y_{i}(e)$ and $\Y_{j}(e)$ exist, and $t,e \models \phi_3$.
Then there exists an \sdpath $\pi$ which addresses a $j$-event in the strict past of $e$
(namely $\pi\cdot\test{\on{j}}(e)$) such that every $i$-event of the form
$\pi'\cdot\test{\on{i}}(e)$ (where $\pi'$ is an \sdpath formula with $1 \le\length{\pi'}
< |\pset|$) lies below $\pi\cdot\test{\on{j}}(e)$.
In particular, $\Y_{i}(e) \leq \pi\cdot\test{\on{j}}(e) \leq Y_{j}(e)$, and this concludes the proof.
\end{proof}

\begin{corollary}\label{cor:comp-with-sdpaths2}
Let $i,j,k \in \pset$ be processes. Then $\Yleq{i,j}{k}$ is logically equivalent to the following $\locpastpdl$ event formula:

\begin{multline*}
\underbrace{\bigvee_{\pi',\pi''} \p{\pi'\cdot\test{\on{i}}\cdot\pi''\cdot\test{\on{j}}}}_{\phi_{1}}
\enspace\land\enspace \underbrace{\bigvee_{\pi} \p{\pi\cdot\test{\on{k}}}}_{\phi_{2}}\\
\land\enspace \underbrace{\left(
\begin{aligned}
\bigvee_{\pi} \bigwedge_{\pi',\pi''} 
\Big(&\p{\pi'\cdot\test{\on{i}}\cdot\pi''\cdot\test{\on{j}}} \enspace\implies \\ &\Pleq{\pi'\cdot\test{\on{i}}\cdot\pi''\cdot\test{\on{j}}}{\pi\cdot\test{\on{k}}}\Big)
\end{aligned} \right)
}_{\phi_{3}}
\end{multline*}
where the disjunctions and conjunctions run over \sdpath formulas $\pi,\pi',\pi''$  of length at least $1$ and at most $|\pset|-1$. 
\end{corollary}
\begin{proof}
 By definition, $t,e \models \Y_{i,j} \leq \Y_{k}$ if and only if $\Y_{j}(\Y_{i}(e))$ and
 $\Y_{k}(e)$ exist, and $\Y_{j}(\Y_{i}(e)) \leq \Y_{k}(e)$.  As in the proof
 of~\cref{cor:comp-with-sdpaths1}, \cref{lem:Y_i and sdpath} can be used to prove that
 $\phi_{1}$ expresses the existence of $\Y_{j}(\Y_{i}(e))$ and $\phi_2$ expresses the
 existence of $\Y_k(e)$.
 
Assume that $t,e \models \Yleq{i,j}{k}$.  Let $\pi$ be given by \cref{lem:Y_i and sdpath}
such that $\Y_{k}(e)= \pi\cdot\test{\on{k}}(e)$, and let $\pi',\pi''$ be arbitrary \sdpath
formulas with length between $1$ and $|\pset|-1$.  If
$\pi'\cdot\test{\on{i}}\cdot\pi''\cdot\test{\on{j}}(e)$ does not exist, then the
$(\pi,\pi',\pi'')$ implication in $\phi_3$ is vacously satisfied.  If it does exist, then
it is a $j$-event which is strictly below an $i$-event in the strict past of $e$.  In
particular, this event sits below $\Y_{j}(\Y_{i}(e)) \leq \Y_{k}(e) =
\pi\cdot\test{\on{k}}(e)$, and hence the corresponding implication in $\phi_3$ is again
satisfied.  Thus $t,e \models \varphi_{3}$.

Conversely, suppose that $\Y_{j}(\Y_{i}(e))$ and $\Y_{k}(e)$ exist and $t,e \models \varphi_3$. Again as in the proof of \cref{cor:comp-with-sdpaths1}, there exists an \sdpath $\pi$ such that the $k$-event $\pi\cdot\test{\on{k}}(e)$ sits above any $j$-event of the form $\pi'\cdot\test{\on{i}}\cdot\pi''\cdot\test{\on{j}}(e)$ (where $\pi',\pi''$ are \sdpath formulas of length between $1$ and $|\pset|-1$).
In particular (again using \cref{lem:Y_i and sdpath}), $\Y_{j}(\Y_{i}(e)) \leq \pi\cdot\test{\on{k}}(e) \le \Y_{k}(e)$ and this concludes the proof.
\end{proof}

\subsubsection{Reducing inequalities to equalities}\label{subsubsec:reduceIneqtoEq}

Our next step is to prove \cref{prop: existence of Pleq}, which relies on the following proposition.

\begin{proposition}\label{prop: existence of Peq}
For each pair of \sdpath formulas $\pi$ and $\pi'$, 
there exists a $\locpastpdl$ event formula $\Peq{\pi}{\pi'}$ such that $t,e \models
\Peq{\pi}{\pi'}$ if and only if $\pi(e)$ and $\pi'(e)$ exist and $\pi(e) = \pi'(e)$.
\end{proposition}

The proof of \cref{prop: existence of Peq} is deferred to \cref{subsubsec:redEqtoSep}. For now, we show how it is used to establish \cref{prop: existence of Pleq}. We first record the following technical lemma.

\begin{lemma}\label{lem:order1}
  Let $t$ be a trace, let $e,f$ be two events in $t$ with $e\not< f$ and let $\pi$ be an 
  \sdpath formula such that $\pi(e)$ exists. Then $\pi(e)\leq f$ if and 
  only if there exists a prefix $\pi'$ of $\pi$ and a \sdpath formula $\pi''$ with 
  $\length{\pi''}\leq|\pset|$ such that $\pi'(e)=\pi''(f)$.
\end{lemma}

\begin{figure}
  \centering
  \begin{tikzpicture}[scale=0.7]

   \node[circle,fill=black,inner sep=0pt,minimum size=5pt,label=below:{$\pi(e)$}] (target) at (3,0) {};
   \node[circle,fill=black,inner sep=0pt,minimum size=5pt,label=below:{$e$}] (e) at (8,5) {};
   \node[circle,fill=black,inner sep=0pt,minimum size=5pt,label=below:{$f$}] (f) at (10,0) {};

   \node[black] at (0,2.5) {$j_{i}$};
   \draw[black, very thick] (0.5,2.5) -- (11.5,2.5);
   \draw[orange, very thick, fill=orange!5] (7,2.30) rectangle (7.25,3.5);
   \node[circle,fill=black,inner sep=0pt,minimum size=3pt,label=below:{$e_{i-1}$}] (e_{i-1}) at (7.12,2.5) {};
   \draw[orange, very thick, fill=orange!5] (4,2.70) rectangle (4.25,1);
   
   \node[circle,fill=black,inner sep=0pt,minimum size=3pt,label={[orange,above]:$\varphi_{i}$}](e_{i}) at (4.12,2.5) {};
   \node[circle,fill=black,inner sep=0pt,minimum size=3pt,label=above:{}] (e_{i2}) at (4.12,1.25) {};
   \node[] at (4.12,0.8) {$e_{i}$};
   
   \node[circle,fill=black,inner sep=0pt,minimum size=3pt,label=above:{}] (e_{(i-1)2}) at (7.12,3.25) {};

   \draw[purple, very thick, fill=purple!5] (6,2.70) rectangle (6.25,1.5);
   \node[circle,fill=black,inner sep=0pt,minimum size=3pt,label=above:{}] (g_{up}) at (6.12,2.5) {};
   \node[circle,fill=black,inner sep=0pt,minimum size=3pt,label=below:{}] (g_{down}) at (6.12,1.75) {};
   \node[] (labelg) at (6.12,0.8) {$\equalto{g}{Y_{j_{i}}(f)}$};

   \path [draw=blue,snake it] (8,5)--(7.12,3.25) (7.8,4) node{$\pi'$};
   \path [draw=blue,snake it] (4.12,1.25)--(3,0);
   
   \path [draw=blue,snake it] (7.12,2.5)--(4.12,2.5);
   \path [draw=red,snake it] (6.12,2.5)--(4.12,2.5);
   \path [draw=red,snake it] (10,0)--(6.12,1.75) (8.5,1.12) node{$\pi'''$};

   \draw[|-|,draw=purple] (4.3,4) -- (6.9,4) node[text =purple,midway, anchor=south]{$\lnot \varphi_{i}$};
\end{tikzpicture}
  \caption{Proof of \cref{lem:order1}}
  \label{fig:enter-label}
\end{figure}

\begin{proof}
One direction is easy: if $\pi'(e)=\pi''(f)$ for a prefix $\pi'$ of $\pi$, then $\pi(e)\leq\pi'(e) = \pi''(f)\leq f$.
 
Conversely, suppose that $\pi(e)\leq f$ with
  $\pi=\test{\varphi'_{0}}\cdot\prevon{j_{1}}{\varphi_{1}}\cdot\test{\varphi'_{1}}
  \cdot\cdots\cdot\test{\varphi'_{k-1}}\cdot\prevon{j_{k}}{\varphi_{k}}\cdot\test{\varphi'_{k}}$ (we may assume that $\pi$ is in standardized form).
  
If $e=f$, the announced result holds with $\pi'=\test{\varphi'_{0}}$ and $\pi''=\test{\True}$. We now assume that $e\neq f$.
  
Let $e_{0}=e$ and, for $1\leq i\leq k$, $e_{i}=\prevon{j_{i}}{\varphi_{i}}(e_{i-1})$. By definition of $\prevon j{\phi}$, we have $j_i \in \loc(e_i)\cap \loc(e_{i-1})$. And we also have $e_k = \pi(e)\leq f$.  Let $i$ be the least index such that $e_{i}\leq f$; then $1\leq i\leq k$ since $e_{0}=e\not\leq f$. Now let $\pi'=\test{\varphi'_{0}}\cdot\prevon{j_{1}}{\varphi_{1}}\cdot\test{\varphi'_{1}}
\cdot\cdots\cdot\test{\varphi'_{i-1}}\cdot\prevon{j_{i}}{\varphi_{i}}$. Then $\pi'$ is a prefix of $\pi$ and $e_{i}=\pi'(e)$.
  
Since $e_i$ and $e_{i-1}$ are $j_{i}$-events, the maximal $j_{i}$-event such that $g\leq f$ satisfies $e_{i}\leq g < e_{i-1}$, see \cref{fig:enter-label}. Moreover, $g = \pi'''(f)$ for some \sdpath formula $\pi'''$ with
$\length{\pi'''}<|\pset|$. Indeed, if $g < f$, we have $g=\Y_{j_{i}}(f)$ (since $e_{i-1}\not\leq f$) and \cref{lem:Y_i and sdpath} shows the existence of $\pi'''$. If instead
$g=f$, then we let $\pi'''=\test{\True}$.
  
We can now conclude the proof: if $g=e_{i}$ then we let $\pi''=\pi'''$.
Otherwise, we have $e_{i}=\prevon{j_{i}}{\varphi_{i}}(g)$ and we let 
$\pi''=\pi'''\cdot\prevon{j_{i}}{\varphi_{i}}$. In either case, we have $\pi'(e)= e_i = \pi''(f)$, as announced. 
\end{proof}

\begin{proof}[Proof of \cref{prop: existence of Pleq}]
Let $\pi,\pi'$ be \sdpath formulas. %
We define $\Pleq{\pi}{\pi'}$ to be the $\locpastpdl$ event formula
$$\Pleq{\pi}{\pi'} = \p{\pi} \wedge \p{\pi'} \wedge \bigvee_{\pi'',\pi'''} \Peq{\pi''}{\pi'\pi'''}$$
where $\pi''$ ranges over prefixes of $\pi$ and $\pi'''$ ranges over \sdpath formulas of
length at most $|\pset|$ (so that $\pi''$ and $\pi'\pi'''$ have length at most
$|\pi|$ and $|\pi'|+|\pset|$, respectively).

First suppose that $t,e \models \Pleq{\pi}{\pi'}$. Clearly $\pi(e)$ and $\pi'(e)$ exist. Since $t, e \models \bigvee\limits_{\pi'',\pi'''} \Peq{\pi''}{\pi'\pi'''}$, there exists a prefix $\pi''$ of $\pi$ and an \sdpath formula $\pi'''$ such that $\pi''(e) = \pi'''(\pi'(e))$. This yields $\pi''(e) \leq \pi'(e)$. Since $\pi''$ is a prefix of $\pi$, it follows that $\pi(e) \leq \pi''(e)$, and hence $\pi(e) \leq \pi'(e)$.

Conversely, suppose that $\pi(e)$ and $\pi'(e)$ exist and that $\pi(e) \leq  \pi'(e)$. The first (existence) condition implies that $t,e \models \p{\pi} \land \p{\pi'}$. 
Since $\pi(e) \leq \pi'(e)$, \cref{lem:order1} applied with $f=\pi'(e)$ shows that
$\pi''(e)=\pi'''(\pi'(e))$ for some prefix $\pi''$ of $\pi$ and an \sdpath formula $\pi'''$
with length at most $|\pset|$.
That is, $t,e \models \Peq{\pi''}{\pi'\pi'''}$ for those particular $\pi'', \pi'''$ and
this completes the proof.
\end{proof}

\subsubsection{Reduction of equalities to separating formulas}\label{subsubsec:redEqtoSep}

Our aim here is to prove \cref{prop: existence of Peq}. For this, we use yet another intermediate result.

\begin{proposition}\label{prop: existence of separating formulas}
For each \sdpath formula $\pi$ of length at least 1, there exists a
finite set $\Xi(\pi)$ of $\locpastpdl$ event formulas such that, for every
trace $t$ and event $e$ in $t$ such that $\pi(e)$ exists, $\Xi(\pi)$ \emph{separates} $e$ and $\pi(e)$: that is, there exist $\xi, \xi'\in \Xi(\pi)$ such that $\xi$ holds at $e$ and not at
$\pi(e)$, and $\xi'$ holds at $\pi(e)$ and not at
$e$.
\end{proposition}

The proof of \cref{prop: existence of separating formulas} is complex and is given in \cref{subsubsec: construction separating}. Let us see immediately how that statement is used to prove \cref{prop: existence of Peq}.

\begin{proof}[Proof of \cref{prop: existence of Peq}]
  Let $\pi,\pi'$ be \sdpath formulas.  Let $\Xi$ be the union of the set $\{\on{i} \mid i
  \in \pset\}$ and of the $\Xi(\sigma)$ where $\sigma$ is an \sdpath formula of length at
  least $1$ and at most $\max(|\pi|,|\pi'|)+|\pset|$. Notice that $\Xi$ is a finite 
  set since, up to equivalence, there are finitely many \sdpath formulas of bounded 
  length. 
  We define the $\locpastpdl$ event formula $\Peq{\pi}{\pi'}$ by
  $$
  \Peq{\pi}{\pi'} = \p{\pi} \land \p{\pi'} \enspace\land\enspace \neg\bigvee_{\xi\in\Xi} 
  \big(\p{\pi}\xi \land \p{\pi'}\neg\xi\big) \,.
  $$

Suppose first that $\pi(e)$ and $\pi'(e)$ exist and are equal. Clearly $t,e \models \p{\pi} \land \p{\pi'}$. Moreover, since $\pi(e) = \pi'(e)$, we have $t,e \models \p{\pi}\varphi$ if and only if $t,e \models \p{\pi'}\varphi$, for every event formula $\varphi$. This implies $t,e \models \neg\bigvee_{\xi\in\Xi}  \big(\p{\pi}\xi \land \p{\pi'}\neg\xi\big)$, and hence $t,e \models \Peq{\pi}{\pi'}$.
  
Conversely, assume that $t,e \models \p{\pi} \land \p{\pi'}$ and let $f = \pi(e)$ and
$f'=\pi'(e)$.  Suppose $f \neq f'$.  We show that $t,e \not\models \Peq{\pi}{\pi'}$, 
i.e., $t,e\models \p{\pi}\xi \land \p{\pi'}\neg\xi$ for some $\xi\in\Xi$.  If
$\loc(f)\not\subseteq \loc(f')$ then we choose $\xi=\on{i}$ for some $i\in\loc(f)
\setminus \loc(f')$. 
  
Suppose now that $\loc(f) \subseteq \loc(f')$. Then $f$ and $f'$ are $\ell$-events for any process $\ell \in \loc(f)$, and hence they are ordered. Without loss of generality, we assume that $f < f'$. 
By Lemma~\ref{lem:order1}, there exist \sdpath formulas $\pi_1, \pi_2, \pi''$ such that $\pi=\pi_{1}\cdot\pi_{2}$, $\length{\pi''} \le |\pset|$ and $\pi_{1}(e) = \pi''(f')$, see \cref{fig:Eq-with-separating}.
\begin{figure}[htb]
    \centering
   \begin{tikzpicture}
\draw[black, very thick] (0.2,1) -- (8,1);
     \filldraw[purple] (7,4) circle (3pt) node[anchor=north west]{$e$};
     \filldraw[purple] (5,1) circle (3pt) node[anchor=north west]{$\pi'(e)$};
     \filldraw[purple] (2,1) circle (3pt) node[anchor=north west]{$\pi(e)$};
     \filldraw[purple] (3.65,2) circle (3pt) node[anchor=south east]{};

\path [draw=blue,snake it] (7,4)--(5,1) (6.5,2.80) node{$\pi'$};
\path [draw=blue,snake it] (7,4)--(3.65,2) (4.5,2.80) node{$\pi_{1}$};
\path [draw=blue,snake it] (3.65,2)--(2,1) (2.35,1.65) node{$\pi_{2}$};
\path [draw=blue,snake it] (3.65,2)--(5,1) (4.75,1.65) node{$\pi''$};
\end{tikzpicture}
    \caption{Proof of \cref{prop: existence of Peq}}
    \label{fig:Eq-with-separating}
\end{figure}
In particular, $f=\pi_{2}(\pi_{1}(e))=\pi_{2}(\pi''(f'))$.
Observe that the path $\pi''\pi_{2}$ has length at most $|\pi|+|\pset|$
and at least $1$ (since $(\pi''\pi_{2})(f')=f<f'$).  \cref{prop: existence of separating
formulas} then implies the existence of a formula $\xi\in\Xi(\pi''\pi_{2})\subseteq\Xi$ 
such that $t,f\models\xi$ and $t,f'\models\neg\xi$.
Therefore $t,e$ satisfies $\p{\pi}\xi \land \p{\pi'}\neg\xi$, and
$t,e \not\models \Peq{\pi}{\pi'}$.
In the symmetric case $f'<f$ we use a formula $\xi'\in\Xi$ such that 
$t,f'\models\neg\xi'$ and $t,f\models\xi'$. We obtain
$t,e\models\p{\pi}\xi' \land \p{\pi'}\neg\xi'$,
which concludes the proof. 
\end{proof}

\subsubsection{Construction of separating formulas}\label{subsubsec: construction separating}

Here we establish \cref{prop: existence of separating formulas}. This is a complex construction which involves constructing, for each \sdpath formula $\pi$, $\locpastpdl$ event formulas $\same{j}{\pi}$ ($j\in\pset$) which capture the fact that both $\p\pi$ and $\p{\prevon j{\p\pi}}$ are defined on a $j$-event $e$, and point to the same event in the past of $e$ (\cref{lm: same is locpastpdl}); and $\locpastpdl$ formulas $\Mod{k,n}{\pi}$ ($0\le k < n$) which allow counting
modulo $n$ the number of events $f$ in the past of the current event such that 
$\pi(f)$ exists and $\pi(f)\neq\pi(f')$ for all $f'<f$ where $\pi$ is defined.

We start with an easy lemma.

\begin{lemma}\label{lem:technical}
Let $\pi$ be an \sdpath formula of length $n$. Let $t$ be a trace and
$e_{0},e_{1},\ldots,e_{n}$ be events in $t$ on each of which $\p\pi$ holds.
If $\pi(e_0) < \pi(e_1) < \cdots < \pi(e_n)$, then $e_{0} \leq \pi(e_n)$.
\end{lemma}

\begin{proof}
The proof is by induction on $n = \length\pi$, and the statement is trivial if $n=0$ 
since in this case $\pi(e_{0})=e_{0}$.

Assume now that $n > 0$, say, $\pi=\test{\varphi'}\cdot\prevon{j}{\varphi}\cdot\pi'$. For each $0\le i\le n$, let $f_i =  \pi(e_i)$. Note that, since $\p\pi$ holds at each $e_i$, then each $e_i$ is a $j$-event, and so is each $e'_{i} = \prevon{j}{\varphi}(e_{i})$. In particular, the $e_i$ and $e'_i$ are totally ordered. In addition, $f_i =\pi'(e'_{i})$, see \cref{fig:lem:technical}.
\begin{figure}[htp]
    \centering
    \begin{tikzpicture}[scale=0.55]
    \node[black] at (0,4) {j};
    \draw[black, very thick] (0.2,1) -- (15,1);
    \draw[black, very thick] (0.2,4) -- (15,4);

    \path [draw=blue,snake it] (2,4)--(1,1) (1,2) node{$\pi'$};
    \path [draw=blue,snake it] (4,4)--(2,4);

    \filldraw[purple] (4,4) circle (3pt)   node[anchor=north]{$e_{0}$};

    \filldraw[brown] (2,4) circle (3pt) node[anchor=south]{$\varphi$}
    (2.25,3.45) node[]{$e'_{0}$};

    \filldraw[orange] (1,1) circle (3pt) node[anchor=north]{$f_{0}$};
    \node[orange] at (2.5,0.5) {<};
    \node[orange] at (5.5,0.5) {<};
    \node[orange] at (10,0.5) {<};

    \path [draw=blue,snake it] (5,4)--(4,1) (4,2) node{$\pi'$};
    \path [draw=blue,snake it] (7,4)--(5,4);

    \filldraw[purple] (7,4) circle (3pt) node[anchor=north]{$e_{1}$};
    
    \filldraw[brown] (5,4) circle (3pt) node[anchor=south]{$\varphi$} (5.25,3.45) node[] {$e'_{1}$};
    \filldraw[orange] (4,1) circle (3pt) node[anchor=north]{$f_{1}$};

    \path [draw=blue,snake it] (12,4)--(11,1) (11,2) node{$\pi'$};
    \path [draw=blue,snake it] (14,4)--(12,4);

    \filldraw[purple] (14,4) circle (3pt) node[anchor=north]{$e_{n}$};
    
    \filldraw[brown] (12,4) circle (3pt) node[anchor=south]{$\varphi$} (12.25,3.45)node[]{$e'_{n}$};
    \filldraw[orange] (11,1) circle (3pt) node[anchor=north]{$f_{n}$};

    \filldraw [fill=purple, draw=black] (9,3.75) rectangle (10,4.25) (9.5,3.5) node {$e_2 \cdots e_{n-1}$};
    \filldraw [fill=orange, draw=black] (8,0.75) rectangle (9,1.25) (8.5,0.45) node[orange] {$f_2 \cdots f_{n-1}$};

    [ultra thick] \draw [draw=red, fill=red, ->] (11,1) -- (5,4);
    \node[text = red] at (7.25,2.25) {$e'_{1} \leq f_{n}$};

\draw[|-|,draw=purple] (2.35,4.5) -- (3.95,4.5) node[text =purple,midway, anchor=south]{$\lnot \varphi$};
\draw[|-|,draw=purple] (5.35,4.5) -- (6.95,4.5) node[text =purple,midway, anchor=south]{$\lnot \varphi$};
\draw[|-|,draw=purple] (12.35,4.5) -- (13.95,4.5) node[text =purple,midway, anchor=south]{$\lnot \varphi$};

\end{tikzpicture}
    \caption{Lemma ~\ref{lem:technical}}
    \label{fig:lem:technical}
\end{figure}

The induction hypothesis implies $e'_1\leq f_n$.  Moreover, $f_0 < f_1$ implies $e'_0 <
e'_1$.  If $e'_1 < e_0$, we have $e'_1 = e'_0$ by definition of $ \prevon{j}{\varphi}$,
a contradiction.  So we have $e_0 \le e'_1 \le f_n$ and this concludes the proof.
\end{proof}

We now establish the existence of the $\locpastpdl$ formulas $\same{j}{\pi}$ announced in the introduction of \cref{subsubsec: construction separating}.

\begin{proposition}\label{lm: same is locpastpdl}
  Let $\pi$ be an \sdpath formula and $j\in\pset$. There exists a $\locpastpdl$ event formula $\same{j}{\pi}$ such that, for every trace $t$ and event $e$, we have $t,e \models \same{j}{\pi}$ if and only if $\p{\pi}$ and $\p{\prevon{j}{\p{\pi}}}$ hold at $e$, and $\pi(e) = \pi(\prevon{j}{\p{\pi}}(e))$.
\end{proposition}

\begin{proof}
The proof is by structural induction on $\pi$. We observe that if $t,e$ satisfies $\p{\prevon{j}{\p{\pi}}}$, then $e$ is a $j$-event.
  
\paragraph*{Case 1: $\pi=\test{\phi}$}
In this case, $\pi(e) = e$ if it exists, whereas $\prevon{j}{\p\pi}$ holds only in the strict past of $e$. So we can never have $\pi(e) = \pi(\prevon{j}{\p{\pi}}(e))$ and we can let $\same{j}{\pi} =\False$.

\paragraph*{Case 2: $\pi=\test{\phi}\cdot\pi'$}

Note that $t,e \models \p\pi$ if and only if $t,e \models \phi \land \p{\pi'}$.
In that case, $\pi(e) = \pi'(e)$.

Assume that $\pi(e)$ and $\prevon{j}{\p\pi}(e)$ exist.
Consider the sequence $e_k < e_{k-1} < \dots < e_0$ of all $j$-events satisfying $\p{\pi'}$, 
starting at $e_k = \prevon{j}{\p\pi}(e)$ and ending at $e_0 = e$, see \cref{fig: testphi dot pi'}.
\begin{figure}[htb]
  \centering
  \begin{tikzpicture}[scale=0.55]
    \node at (0,5) {j};
    \draw[black, very thick] (0.2,1) -- (15,1);
    \draw[black, very thick] (0.2,5) -- (15,5);
        
    \filldraw[purple] (14,5) circle (3pt) node[anchor=north west]{$e = e_{0}$} node [anchor=south] {$\varphi$};    
    \filldraw[purple] (2,5) circle (3pt) node[anchor=south]{$\varphi$} node[anchor=north west] {$e_{k}$};
    \filldraw[orange] (1,1) circle (3pt) node[anchor=north]{$\pi(e)$};
  
    \filldraw[purple] (10,5) circle (3pt) node[anchor=north west]{$e_{1}$} node [anchor = south]{$\lnot \varphi$};
    \filldraw[purple] (4,5) circle (3pt) node[anchor=north west]{$e_{k-1}$} node [anchor=south] {$\lnot \varphi$};

  \draw[|-|,draw=purple] (2.2,6) -- (13.8,6) node[text =purple,midway, anchor=south]{$( \langle\pi'\rangle \implies (\lnot \varphi \land same_{j}(\pi'))$};
    \path [draw=blue,snake it] (2,5)--(1,1) (1,2.80) node{$\pi'$};
    \path [draw=blue,snake it] (14,5)--(1,1) (7,2.5) node{$\pi'$};
    \path [draw=blue,snake it] (10,5)--(1,1) (4,2.75) node{$\pi'$};
    \path [draw=blue,snake it] (4,5)--(1,1) (2.9,2.80) node{$\pi'$};
  \filldraw [fill=purple, draw=black] (6,4.75) rectangle (8,5.25) (6.75,4.5) node[purple] {$e_{k-2} \cdots  e_{2}$};
  \end{tikzpicture}
  \caption{Case $\pi=\test{\phi}\cdot\pi'$}  
    \label{fig: testphi dot pi'}
\end{figure}
In particular, $e_{i+1} = \prevon{j}{\p{\pi'}}(e_i)$ for every $0 \le i < k$.

If $\pi(e) = \pi(\prevon{j}{\p{\pi}}(e))$, then $\pi'$ also maps $e$,
$\prevon{j}{\p{\pi}}(e)$, and hence all the $e_i$, to the same event: equivalently,
$\pi'(e_{i+1}) = \pi'(e_i)$ for every $0\leq i < k$.

Conversely, suppose that, for every $0\leq i < k$, we have $\pi'(e_i) = \pi'(e_{i+1})$.
It follows that every $\pi'(e_i)$ is equal to $\pi'(e_0)$.  By definition, $\pi'(e_0) =
\pi(e)$ and $\pi'(e_k) = \pi(e_k)=\pi(\prevon{j}{\p{\pi}}(e))$.  Therefore,
$\prevon{j}{\p\pi}(e) = \pi(e)$.

Thus, if $\pi(e)$ and $\prevon{j}{\p{\pi}}(e)$ exist, then $\same{j}{\pi}$ holds at $e$ if
and only if the $j$-events strictly between $\prevon{j}{\p\pi}(e)$ and $e$ that do satisfy
$\p{\pi'}$ (that is: the $e_i$) do not satisfy $\phi$, and do satisfy $\same{j}{\pi'}$.
This justifies letting $\same{j}{\pi}$ be %
$$\varphi\wedge\same{j}{\pi'}\wedge
\left\langle \lmove_{j}\cdot
\Big(\test{\big(\neg\p{\pi'}\lor(\neg\varphi\land\same{j}{\pi'})\big)}\cdot\lmove_{j}\Big)^{*}
\right\rangle \p{\pi} .$$

\paragraph*{Case 3: $\pi=\prevon{i}{\varphi}\cdot\pi'$}
An event $e$ satisfies $\p\pi$ if and only if $e' = \prevon{i}{\phi}(e)$
exists and satisfies $\p{\pi'}$.  In that case, $\pi(e) = \pi'(e')$. We first consider the case where $j=i$, which is notationally slightly simpler yet contains the substance of the proof.

Let $e''$ be the immediate predecessor of $e$ on process $i$. 
Suppose first that $e''$ does not satisfy $\phi$
(see \cref{fig: case 3.1 e'' not phi}).
    \begin{figure}
      \centering
      \begin{tikzpicture}[scale=0.5]
        \node at (0,5) {i=j};
        \draw[black, very thick] (0.5,5) -- (15,5);

        \filldraw[purple] (14,5) circle (3pt) node[anchor=north west]{$e$} node [anchor=south] {$\varphi$};    
        \filldraw[purple] (2,5) circle (3pt) node[anchor=south]{$\varphi$} node[anchor=north west] {$e'$};
        \filldraw[orange] (1,1) circle (3pt) node[anchor=north]{$\pi(e)$};
    
        \filldraw[purple] (12.2,5) circle (3pt) node[anchor=north]{$\equalto{e''}{\xleftarrow[]{}_{i}(e)}$} node [anchor = south]{$\lnot \varphi$};

    \draw[|-|,draw=purple] (2.2,6) -- (13.8,6) node[text =purple,midway, anchor=south]{$\lnot \varphi$};
        \path [draw=blue,snake it] (2,5)--(1,1) (1,2.80) node{$\pi'$};
         \path [draw=blue,snake it] (14,5)--(2,5);
    \end{tikzpicture}
      \caption{Case $\pi=\prevon{i}{\varphi}\cdot\pi'$ and $t,e'' \models  \lnot \varphi$}
      \label{fig: case 3.1 e'' not phi}
    \end{figure}
Then $e' < e''$ and $e' = \prevon{i}{\phi}(e'')$.  It follows that $e''$ satisfies
$\p\pi$, $e'' = \prevon{i}{\p{\pi}}(e)$ and $\pi(e'') = \pi'(e') = \pi(e)$.

Suppose now that $e''$ satisfies $\phi$ (see \cref{fig: case 3.1 e'' phi}).
    \begin{figure}
      \centering
      \begin{tikzpicture}
        \node at (0,4) {i=j};
        \draw[black, very thick] (0.5,4) -- (8,4);

        \draw (1,4) to (1,4.75);
        \filldraw[orange] (0.5,0.5) circle (2pt) node[anchor=north]{$\pi(e)$};
        \filldraw[purple] (1,4) circle (2pt) node[anchor=north east]{$f'$} to (1,4.75) node[anchor=south]{$\varphi$} (1,5.25) node[]{$\langle\pi'\rangle$};
        \filldraw[purple] (7,4) circle (2pt) node[anchor=north]{$e'=e''$} to (7,4.75) node[anchor=south]{$\varphi$} (7,5.25) node[]{$same_{i}(\pi')$};
        \filldraw[purple] (5.5,4) circle (2pt) node[anchor=north]{} to (5.5,4.75) node[anchor=south]{$\lnot \varphi$} (5.5,5.25) node[]{$same_{i}(\pi')$};
        \filldraw[purple] (4,4) circle (2pt) node[anchor=north]{} to (4,4.75) node[anchor=south]{$\lnot \varphi$} (4,5.25) node[]{$same_{i}(\pi')$};
        \filldraw[purple] (2.5,4) circle (2pt) node[anchor=north]{} to (2.5,4.75) node[anchor=south]{$\lnot \varphi$} (2.5,5.25) node[]{$same_{i}(\pi')$};
         \filldraw[purple] (7.8,4) circle (2pt) node[anchor=north west]{$e$} to (7.8,5.5) node[anchor=south]{$same_{i}(\pi)$};

        \path [draw=blue,snake it] (0.5,0.5)--(1,4) (0.5,2.50) node{$\pi'$};
       \path [draw=blue,snake it] (0.5,0.5)--(7,4) (1.5,3.5) node{$=$};
       \path [draw=blue,snake it] (0.5,0.5)--(5.5,4) (2.75,3.50) node{$=$};
       \path [draw=blue,snake it] (0.5,0.5)--(4,4) (4,3.50) node{$=$};
       \path [draw=blue,snake it] (0.5,0.5)--(2.5,4) (5.5,3.50) node{$=$} (4.5,2.5) node {$\pi'$};
    
    \draw[|-|,draw=purple] (1.1,4.1) -- (2.4,4.1) node[text =purple,midway, anchor=south]{$\lnot \langle\pi'\rangle$};
    \draw[|-|,draw=purple] (2.6,4.1) -- (3.9,4.1) node[text =purple,midway, anchor=south]{$\lnot \langle\pi'\rangle$};
    \draw[|-|,draw=purple] (4.1,4.1) -- (5.4,4.1) node[text =purple,midway, anchor=south]{$\lnot \langle\pi'\rangle$};
    \draw[|-|,draw=purple] (5.6,4.1) -- (6.9,4.1) node[text =purple,midway, anchor=south]{$\lnot \langle\pi'\rangle$};
       
    \end{tikzpicture}
      \caption{Case $\pi=\prevon{i}{\varphi}\cdot\pi'$ and $t,e'' \models \phi$}
      \label{fig: case 3.1 e'' phi}
    \end{figure}
Then $e' = e''$.  Assume that $f = \prevon{i}{\p\pi}(e)$ exists.  Then $f\leq e'$ and $f$
satisfies $\p{\pi}$, so $f' = \prevon{i}{\phi}(f)$ exists and satisfies $\p{\pi'}$.

We first verify that $f' = \prevon{i}{\phi \land \p{\pi'}}(e')$: if it is not the case,
there exists an $i$-event $g$ strictly between $f'$ and $e'$, satisfying $\phi \land
\p{\pi'}$.  Then the immediate successor of $g$ on the $i$-line satisfies $\p{\pi}$.
This yields a contradiction since $g > f' =
\prevon{i}{\varphi}(\prevon{i}{\p{\pi}}(e))$.

Let now $f' = e_k < \dots < e_1 = e'$ be the sequence of $i$-events between $f'$ and $e'$
which
satisfy $\p{\pi'}$.  In particular, for each $1\le h < k$, $e_{h+1} =
\prevon{i}{\p{\pi'}}(e_h)$.

Now $\pi(e) = \pi(f)$ if and only if $\pi'(e') = \pi'(f')$. This is equivalent to requiring that $\pi'(e_h) = \pi'(e') = \pi'(e_1)$ for all $h$, which in turn is equivalent to each $e_h$ ($1\le h < k$) satisfying $\same{i}{\pi'}$. Since $f' = \prevon{i}{\phi \land \p{\pi'}}(e')$ note that all events in the interval $(f',e')$ either satisfy $\neg\p{\pi'}$ or they are among $e_h$ hence they satisfy $\neg\phi \land \same{i}{\pi'}$

This justifies letting $\same{j}{\pi}$ be the conjunction of $\p{\pi}$ and
  \begin{align*}
\p{\lmove_{i}} &\neg\phi \lor \Big\langle \lmove_{i} \cdot\test{(\varphi \land \same{i}{\pi'})} \cdot \lmove_{i} \cdot \\
  & \Big( \test{\big( \p{\pi'} \implies (\neg\varphi \wedge \same{i}{\pi'}) \big)} \cdot \lmove_{i} \Big)^{*}
  \cdot \test{\big( \p{\pi'}\wedge\varphi \big)} 
  \Big\rangle.
  \end{align*}

\paragraph*{Sub-case 3.2: $i \ne j$}
Here we let $\same{j}{\pi} = \same{i}{\on{j}?\cdot\pi}$. Note that Case 2 gives a formula for $ \same{i}{\on{j}?\cdot\pi}$ in terms of $\same{i}{\pi}$, which in turn is constructed in Case 3.1 above.

To justify this choice for $\same{j}{\pi}$, we observe that any event satisfying $\p\pi$
is an $i$-event, since $\pi$ starts with $\lmove_i$.  Similarly, if $\prevon{j}{\p{\pi}}$ holds at an event $e$, then both $e$ and $f=\prevon{j}{\p{\pi}}(e)$ are $j$-events (since $\prevon{j}{\p{\pi}}$ starts and ends with $\lmove_j$)%
.  As a result, $\pi$ and $\prevon{j}{\p{\pi}}$ are defined at $e$ if
and only if $\test{\on{j}}\cdot\pi$ and $\prevon{i}{\p{\test{\on{j}}\cdot\pi}}$ are. Moreover, in that case, we have 
$f = \prevon{j}{\p{\pi}}(e) = \prevon{i}{\p{\test{\on{j}}\cdot\pi}}(e)$, $\pi(e) = (\test{\on{j}}\cdot\pi)(e)$ and 
$\pi(f) = (\test{\on{j}}\cdot\pi)(f)$. Therefore $\pi(e) = \pi(f)$ if and only if $(\on{j}?\cdot\pi)(e) = (\on{j}?\cdot\pi)(f)$. This concludes the proof.
\end{proof}

We now use Lemma~\ref{lm: same is locpastpdl} to construct, for each \sdpath formula $\pi$, $\locpastpdl$ event formulas which count modulo $n$ the number of events $f$ in the past of the current event $e$ with different $\pi$-images.

\begin{lemma}\label{lem:path-mod}
  Let $\pi$ be an \sdpath formula of positive length, say $\pi = \test{\phi'}\cdot\prevon{i}{\phi}\cdot\pi'$. For every $n>1$ and $0\leq k<n$, there exists a $\locpastpdl$ 
  event formula $\Mod{k,n}{\pi}$ such that, for any trace $t$ and event $e$ in $t$,
  $t,e\models\Mod{k,n}{\pi}$ if and only if $t,e\models\p{\pi}$ and
  $$
  \left|\Big\{f<e \mid t,f\models\p{\pi}\wedge\neg\same{i}{\pi}\Big\}\right| \enspace=\enspace k\pmod n.
  $$
\end{lemma}

\begin{proof}
We note that every event satisfying $\p{\pi}$ is an $i$-event. Let $\psi=\p{\pi}\land\neg\same{i}{\pi}$.  

If $t,e \models \p\pi$, then the events $f < e$ satisfying $\psi$ are the events
$(\prevon{i}{\psi})^m(e)$ ($m > 0$), and the minimal such event (which does not satisfy
$\prevon{i}{\psi}$) satisfies $\neg\big(\p{\lmove_i^+}\psi\big)$, see \cref{fig: mod counting lemma}.
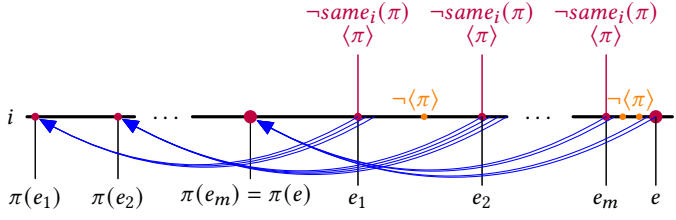
\begin{figure}[htb]
  \centering
  \begin{tikzpicture}[line cap = round, line join = round, >=triangle 45,scale=1.1]
    \node at (0,2.5) {$i$};
    \draw[black, very thick] (0.2,2.5) -- (1.5,2.5);
    \draw[black, very thick] (2.2,2.5) -- (6,2.5);
    \draw[black, very thick] (6.8,2.5) -- (8,2.5);
     
      \filldraw[purple] (4.2,2.5) to (4.2,3.25) node[anchor=south]{$\p{\pi}$} (4.2,3.75) node[]{$\lnot same_{i}(\pi)$};
      \filldraw[black] (4.2,2.5) to (4.2,1.67) node[anchor=north]{$e_{1}$};

      \filldraw[purple] (7.2,2.5) to (7.2,3.25) node[anchor=south]{$\p{\pi}$} (7.2,3.75) node[]{$\lnot same_{i}(\pi)$};
      \filldraw[black] (7.2,2.5) to (7.2,1.67) node[anchor=north]{$e_{m}$};
      \node at (6.4,2.46) {$\cdots$};

      \filldraw[black] (0.3,2.5) to (0.3,1.75) node[anchor=north]{$\pi(e_{1})$};
      \filldraw[black] (1.3,2.5) to (1.3,1.75) node[anchor=north]{$\pi(e_{2})$};
      \filldraw[black] (2.9,2.5) to (2.9,1.75); 
      \node at (2.85,1.55) {$\pi(e_{m})=\pi(e)$};

      \node[big dot] (pi1) at (0.3,2.5) {};
      \node[big dot] (e1) at (4.2,2.5) {};
      \draw[->, blue] (4.2,2.5) to[bend left] (pi1);
      \draw[->, blue] (4.3,2.5) to[bend left] (pi1);
      \draw[->, blue] (4.4,2.5) to[bend left] (pi1);

      \node[bigorange dot] () at (5,2.5) {};
      \node [orange] (notpi) at (4.9,2.7) {$\lnot\p{\pi}$};

      \node[big dot] (pi3) at (1.3,2.5) {};
      \node[big dot] (e3) at (5.7,2.5) {};
      \node at (1.9,2.46) {$\cdots$};
      \draw[->, blue] (5.7,2.5) to[bend left] (pi3);
      \draw[->, blue] (5.8,2.5) to[bend left] (pi3);
      \draw[->, blue] (5.9,2.5) to[bend left] (pi3);
      \draw[->, blue] (6.0,2.5) to[bend left] (pi3);

      \filldraw[purple] (5.7,2.5) to (5.7,3.25) node[anchor=south]{$\p{\pi}$} (5.7,3.75) node[]{$\lnot same_{i}(\pi)$};
      \filldraw[black] (5.7,2.5) to (5.7,1.67) node[anchor=north]{$e_{2}$};

      \node[bigbig dot] (piem) at (2.9,2.5) {};
      \node[big dot] (em) at (7.2,2.5) {};
      \node[bigbig dot] (e) at (7.8,2.5) {};

      \filldraw[black] (7.8,2.5) to (7.8,1.67) node[anchor=north] {$e$};

      \node[bigorange dot] () at (7.4,2.5) {};
      \node[bigorange dot] () at (7.6,2.5) {};

      \node [orange] (notpi) at (7.5,2.7) {$\lnot\p{\pi}$};

      \draw[->, blue] (7.2,2.5) to[bend left] (piem);
      \draw[->, blue] (7.3,2.5) to[bend left] (piem);
      \draw[->, blue] (7.7,2.5) to[bend left] (piem);
      \draw[->, blue] (7.8,2.5) to[bend left] (piem);

\end{tikzpicture}
  \caption{Mod Counting}  
  \label{fig: mod counting lemma}
\end{figure}
As a result, choosing the following formulas for the $\Mod{k,n}{\pi}$
\begin{align*}
\Mod{0,n}{\pi} & = \p{\pi} \land
    \Big\langle \Big( \big( \prevon{i}{\psi} \big)^{n} \Big)^{*} \Big\rangle
    \neg\big(\p{\lmove_{i}^{+}}\psi\big)
\\
    \Mod{k,n}{\pi} & = \p{\pi} \land \p{\prevon{i}{\psi}} \Mod{k-1,n}{\pi}
    \text{ for every } 1\leq k<n.
\end{align*}
proves the statement.
\end{proof}

We can finally complete the task of \cref{subsec:eliminateEventConstants}.

\begin{proof}[Proof of \cref{prop: existence of separating formulas}]
Let $\pi$ be an \sdpath formula of positive length, say $\pi = \test{\phi'}\cdot\prevon{i}{\phi}\cdot\pi'$, and let $n=\length{\pi}+1$. In particular, $\p\pi$ holds only at $i$-events. Let
  \begin{multline*}
    \Xi(\pi)\enspace=\enspace\Big\{\lnot \p{\pi},\p{\pi},\same{i}{\pi},\neg\same{i}{\pi}\Big\} \\ \enspace\cup\enspace
  \Big\{\Mod{k,n}{\pi}\mid 0\leq k<n\Big\} \,.
  \end{multline*}
We want to show that $\Xi(\pi)$ is a separating set of formulas for $\pi$, that is, if $t$
is a trace, $e$ is an event and $\pi(e)$ exists, then there exist $\xi \in \Xi(\pi)$ such
that $\xi$ holds at $e$ and not at $\pi(e)$. It follows easily from the definition of
$\Xi(\pi)$ that there is also a formula $\xi'\in\Xi(\pi)$ such
that $\xi'$ holds at $\pi(e)$ and not at $e$.

Suppose that this is not the case: there exists a trace $t$ and an event $e$ such that
$\pi(e)$ exists and, for every formula $\xi\in \Xi(\pi)$, if $\xi$ holds at $e$, then it
holds at $\pi(e)$ as well.  Since $\p\pi \in \Xi(\pi)$ holds at $e$, the event $\pi(e)$
also satisfies $\p\pi$.

Let $e_{1} < \cdots < e_{m}$ be the $i$-events $f$ such that $\pi(e) < f \le e$ and
$t,f\models\p{\pi}\land\neg\same{i}{\pi}$.  If $m = 0$, then every $i$-event $f$ such that
$\pi(e) < f \le e$ which satisfies $\p\pi$, satisfies $\same{i}{\pi}$ as well.  In
particular, $e$ satisfies $\same{i}{\pi}$ and it follows that $\pi(\pi(e)) = \pi(e)$.
This cannot be since $\length\pi > 0$.  Therefore $m > 0$.

By construction, $\pi(e) < e_1$ and $e_{h-1} \le \prevon{i}{\p\pi}(e_h)$ for each $1 < h \le m$. By definition of $\same{i}{\pi}$, we have $\pi(e_{1})<\cdots<\pi(e_{m}) \le \pi(e) < e_1$.
If $m\geq n = \length\pi +1$, then Lemma~\ref{lem:technical} yields $e_{1}\leq\pi(e_{m})$, a contradiction. Therefore, $1 \le m < n$.

Now consider the formula $\same{i}{\pi}$ and $\neg\same{i}{\pi}$ in $\Xi(\pi)$. Assuming, as we do, that $\Xi(\pi)$ is not separating, shows that $e$ satisfies $\same{i}{\pi}$ if and only if $\pi(e)$ does. Therefore $e$ satisfies $\p\pi \land \neg\same i\pi$ if and only $\pi(e)$ does. As a result, $m$ is also the number of events $f$ satisfying $\p\pi \land \neg\same i\pi$ and $\pi(e)\leq f<e$.

Finally, let $k$ be such that $t,e\models\Mod{k,n}{\pi}$ (there is necessarily such a $k$). Since $\Mod{k,n}{\pi} \in \Xi(\pi)$, $\pi(e)$ too satisfies $\Mod{k,n}{\pi}$, and this implies that $m=0\mod n$ --- again a contradiction since we verified that $1 \le m < n$. This concludes the proof that $\Xi(\pi)$ is a separating set.
\end{proof}%

\subsection{Expressing Constant Trace Formulas in $\locpastpdl$}\label{subsec:eliminateTraceConstants}

This subsection is dedicated to proving the analog of \cref{Yconstants theorem} for constant trace formulas.

\begin{theorem}\label{Lconstants theorem}
  The constant trace formulas $\Lleq{i}{j}$ and $\Lleq{i,j}{k}$ can be expressed in $\locpastpdl$
\end{theorem}

The proof strategy is the same as for \cref{Yconstants theorem}. We start by strengthening \cref{lem:order1}, using the fact that we are now only concerned with events at the end of the trace.

\begin{lemma}\label{lem:prefix-maximum}
Let $i$ be a process, $t$ be a trace, $e$ be the maximum $i$-event of $t$ and $f$ be any event. Let $\pi$ be an \sdpath formula such that $\pi(e)$ is defined. Then $\pi(e) \leq f$ if and only if $\pi$ can be factored as $\pi = \pi_{1}\cdot\pi_{2}$ and there exists an \sdpath formula $\pi_{3}$ of length at most $|\pset|$ such that $\pi_{1}(e) = \pi_{3}(f)$.         
\end{lemma}

\begin{proof}
  If $\pi_1, \pi_2, \pi_3$ exist such that $\pi =\pi_{1} \cdot \pi_{2}$ and $\pi_{3}(f)
  =\pi_{1}(e)$ then concatenating both sides by $\pi_{2}$ we have $\pi(e) =
  \pi_2(\pi_{1}(e)) \le \pi_{1}(e) = \pi_{3}(f) \le f$.

Conversely, suppose that $\pi(e) \le f$.  If $e \nless f$, we apply \cref{lem:order1} to
get the expected result.  Suppose now that $e < f$.  Since $e$ is the maximum $i$-event,
\cref{lem:Y_i and sdpath} shows that there exists an \sdpath formula $\pi'$ of length at
most $|\pset|$ such that $e = \pi'(f)$.  We conclude by letting $\pi_3 = \pi'$, $\pi_1 =
\test\True$ and $\pi_2 = \pi$.
\end{proof}

\begin{lemma}\label{lem: LEQ EQ pi1 pi2}
Let $i, j$ be processes and let $\pi_1, \pi_2$ be \sdpath formulas.  There exist
$\locpastpdl$ trace formulas $\EQ{i}{j}{\pi_{1}}{\pi_{2}}$ and
$\LEQ{i}{j}{\pi_{1}}{\pi_{2}}$ such that for all traces $t$, $t$ satisfies
$\EQ{i}{j}{\pi_{1}}{\pi_{2}}$ (resp.\ $\LEQ{i}{j}{\pi_{1}}{\pi_{2}}$) if and only if
$\pi_{1}(\Li)$ and $\pi_{2}(\Lj)$ exist, and $\pi_{1}(\Li) = \pi_{2}(\Lj)$ (resp.\
$\pi_{1}(\Li) \leq \pi_{2}(\Lj)$).
\end{lemma}

\begin{proof}
Let $\EQ{i}{j}{\pi_{1}}{\pi_{2}}$ and  $\LEQ{i}{j}{\pi_{1}}{\pi_{2}}$ be the following formulas:
\begin{align*}
  \EQ{i}{j}{\pi_{1}}{\pi_{2}} = &\EM_{i}(\p{\pi_{1}}) \land \EM_{j}(\p{\pi_{2}}) \\
  &\land \lnot \bigvee_{\xi \in \Xi} \EM_{i}(\langle \pi_{1} \rangle \xi) \land \EM_{j}(\langle \pi_{2} \rangle \lnot \xi) \\
  \LEQ{i}{j}{\pi_{1}}{\pi_{2}} =& \EM_{i}(\p{\pi_{1}}) \land \EM_{j}(\p{\pi_{2}}) \\
  & \land \bigvee_{\pi_{1}', \pi_{3}} \EQ{i}{j}{\pi_{1}'}{\pi_{2}\cdot\pi_{3}} 
\end{align*}
with $\Xi = \{\test{\on{i}} \mid i \in \pset\} \cup \bigcup_{\pi} \Xi(\pi)$ where the union runs over the \sdpath formulas  $\pi$ with length at least $1$ and at most $|\pset| + \max(||\pi_{1}||,||\pi_{2}||)$ (and $\Xi(\pi)$ is given by \cref{prop: existence of separating formulas}); and where the last disjunction is over the prefixes $\pi_{1}'$ of $\pi_{1}$ and the \sdpath formulas $\pi_{3}$ of length at most $|\pset|$.

Observe first that $t \models \EM_{i}(\p{\pi_{1}}) \land \EM_{j}(\p{\pi_{2}})$ if and only if $\pi_{1}(\Li)$, $\pi_{2}(\Lj)$ exist. If in addition $\pi_{1}(\Li) = \pi_{2}(\Lj)$, then, for every event formula $\varphi$, we have $t \models \EM_{i}(\p{\pi_{1}}\varphi)$ if and only if $t \models \EM_{j}(\p{\pi_{2}}\varphi)$. In particular, $t$ satisfies $\EQ{i}{j}{\pi_{1}}{\pi_{2}}$.

Conversely, suppose that $f_1 =\pi_{1}(\Li)$ and $f_2= \pi_{2}(\Lj)$ exist and $f_1 \neq f_2$. If $\loc(f_1) \nsubseteq \loc(f_2)$, let $\ell$ be a process in $\loc(f_1)\setminus\loc(f_2)$. Then $\xi = \test{\on{\ell}}\in \Xi$ and $t$ satisfies $\EM_{i}(\langle \pi_{1} \rangle \xi) \land \EM_{j}(\langle \pi_{2} \rangle \lnot \xi)$. Thus $t$ does not satisfy $\EQ{i}{j}{\pi_{1}}{\pi_{2}}$.

If instead $\loc(f_1) \subseteq \loc(f_2)$, let $\ell\in \loc(f_1)$. Then $f_1$ and $f_2$ are $\ell$-events, and hence $f_1 < f_2$ or $f_2 < f_1$. Assume $f_1 < f_2$ (the proof is similar if $f_2 < f_1$). By \cref{lem:prefix-maximum}, there exist \sdpath formulas $\pi'_1, \pi''_1, \pi_3$ such that $\pi_{1} = \pi_{1}'\cdot\pi_{1}''$, $\length{\pi_3} \le |\pset|$ and $\pi_{1}'(\Levent{i}{t}) = \pi_{3}(f_2)$, see \cref{fig: case in 3.18}.
\begin{figure}[h]
  \centering
  \begin{tikzpicture}[scale=0.94]
    \draw[black, very thick] (0.2,2.5) -- (4,2.5);
         \filldraw[purple] (2.5,4) circle (3pt) node[anchor=north west]{$e = \Li$};
         \filldraw[purple] (2.5,1) circle (3pt) node[anchor=north west]{$\Lj$};
         \filldraw[purple] (1,2.5) circle (3pt) node[anchor=north east]{$\pi_{1}(\Li)$};
         \filldraw[purple] (1.8,3.3) circle (3pt) node[anchor=south east]{};
         \filldraw[purple] (2,2.5) circle (3pt) (2.8,2.3) node[]{$\pi_{2}(\Lj)$};
    
    \path [draw=blue,snake it] (2.5,4)--(1.8,3.3) (2,3.9) node{$\pi_{1}'$};
    \path [draw=blue,snake it] (1.8,3.3)--(1,2.5) (1.1,3) node{$\pi_{1}''$};
    \path [draw=blue,snake it] (2.5,1)--(2,2.5) (2.6,1.65) node{$\pi_{2}$};
    \path [draw=blue,snake it] (2,2.5)--(1.8,3.3) (2.1,3) node{$\pi_{3}$};

    \draw[black, very thick] (4.5,4) -- (4.5,1);

    \draw[black, very thick] (5.2,2.5) -- (9,2.5);

    \filldraw[purple] (7.5,4) circle (3pt) node[anchor=north west]{$\Li$};
    \filldraw[purple] (7.5,1) circle (3pt) node[anchor=north west]{$\Lj$};
    \filldraw[purple] (6,2.5) circle (3pt) node[anchor=north east]{$\pi_{2}(\Lj)$};
    \filldraw[purple] (6.8,1.75) circle (3pt) node[anchor=south east]{};
    \filldraw[purple] (7,2.5) circle (3pt) (7.8,2.3) node[]{$\pi_{1}(\Li)$};

    \path [draw=blue,snake it] (7.5,4)--(7,2.5) (7,3.3) node{$\pi_{1}$};
    \path [draw=blue,snake it] (6.8,1.75)--(6,2.5) (6.2,2) node{$\pi_{2}''$};
    \path [draw=blue,snake it] (7.5,1)--(6.8,1.75) (7,1) node{$\pi_{2}'$};
    \path [draw=blue,snake it] (7,2.5)--(6.8,1.75) (7.2,2) node{$\pi_{3}$};

\end{tikzpicture}
  \caption{Applying \cref{lem:prefix-maximum} when $e = \Li$, $f_1 = \pi_{1}(\Li)$, $f_2 = \pi_{2}(\Li)$ and $f_1 < f_2$ (left) or $f_2 < f_1$ (right)}
  \label{fig: case in 3.18}  
\end{figure}
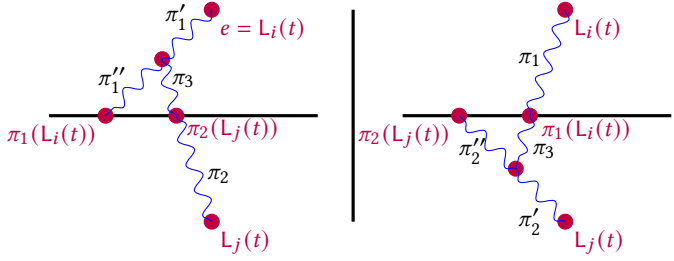
In particular $f_{1} = \pi_{1}''(\pi_{3}(f_2))$.  Let then $\pi' = \pi_{3}\cdot\pi_{1}''$.  In
particular $\length{\pi'} \ge 1$ (since $f_1=\pi'(f_2)<f_2$) and $\length{\pi'_1} \le
\length{\pi_1}$, so $\Xi(\pi')\subseteq \Xi$.
Therefore, we find $\xi\in \Xi(\pi')$ such that $t$ satisfies $\EM_{i}(\langle \pi_{1}
\rangle \xi) \land \EM_{j}(\langle \pi_{2} \rangle \lnot \xi)$, and hence it does not
satisfy $\EQ ij{\pi_{1}}{\pi_{2}}$.

This concludes the proof of the statement on $\EQ ij{\pi_{1}}{\pi_{2}}$. The proof of the statement on $\LEQ ij{\pi_{1}}{\pi_{2}}$ is a direct application of \cref{lem:prefix-maximum}, which states that $\pi_{1}(\Li) \leq \pi_{2}(\Lj)$ if and only if there exist \sdpath formulas $\pi'_1$, $\pi''_1$ and $\pi_3$ such that $\pi_{1} = \pi_{1}'\cdot\pi_{1}''$, $\length{\pi_{3}}\le |\pset|$ and $\pi_{1}'(\Li) = \pi_{3}(\pi_{2}(\Lj))$. This is true if and only if $t$ satisfies $\bigvee_{\pi_{1}', \pi_{3}} \EQ{i}{j}{\pi_{1}'}{\pi_{2}\cdot\pi_{3}}$. 
\end{proof}

We can now complete the proof of the expressive completeness of $\locpastpdl$.

\begin{proof}[Proof of \cref{Lconstants theorem}]
We want to show that the constant trace formulas $\Lleq{i}{j}$ and $\Lleq{i,j}{k}$ are logically equivalent to $\locpastpdl$ formulas.

It is immediate that $\Lleq{i}{j}$ is equivalent to the $\locpastpdl$ formula $\LEQ{i}{j}{\True?}{\True?}$ from \cref{lem: LEQ EQ pi1 pi2}.

Now recall that $t$ satisfies $\Lleq{i,j}{k}$ if and only if $\Levent{k}{t}$,
$\Levent{i,j}{t}$ exist and $\Levent{i,j}{t} \leq \Levent{k}{t}$. 

Notice that $\Levent{i}{t}$ is a $j$-event if and only if $t \models \EM_{i}(\on j)$.  In
this case, $\Levent{i,j}{t} = \Levent{i}{t}$.  If $\Levent{i}{t}$ is not a $j$-event, then
$\Levent{i,j}{t}$ is equal to $\Y_{j}(\Levent{i}{t})$ and \cref{lem:Y_i and sdpath}
asserts that $\Levent{i,j}{t} = (\pi\cdot\test{\on{j}})(\Li)$ for some \sdpath formula $\pi$
of length at least 1 and at most $|\pset|-1$.

The existence of $\Levent{i,j}{t}$ and $\Levent{k}{t}$ is asserted by the sentence
$\bigvee_{\pi}\EM_{i}(\p{\pi}\on{j}) \land \EM_{k}(\True)$, where the disjunction runs
over the \sdpath formulas $\pi$ of length at most $|\pset|-1$.

This justifies considering the following $\locpastpdl$ formula:
\begin{align*}
  \Phi = & \EM_{k}(\True) \wedge \bigvee_{\pi}\EM_{i}(\p{\pi}\on{j}) \\
  & \wedge\left(\bigwedge_{\pi}\EM_{i}(\p{\pi}\on{j}) \Longrightarrow 
  \LEQ{i}{k}{\pi\cdot\test{\on{j}}}{\test{\True}}\right),
\end{align*}
where both the disjunction and the conjunction run over the \sdpath formulas $\pi$ of
length at most $|\pset|-1$.
As discussed, if $t$ satisfies $\Phi$, then it satisfies $\Lleq{i,j}{k}$.

Conversely, suppose that $t$ satisfies $\Lleq{i,j}{k}$ and that $\Levent{k}{t}$ and
$\Levent{i,j}{t}$ exist.  Then every event of the form $(\pi\cdot\test{\on{j}})(\Li)$ lies
below $\Levent{i,j}{t}$, and hence below $\Levent{k}{t}$.  Thus $t$ satisfies $\Phi$.
This establishes that $\Lleq{i,j}{k}$ is logically equivalent to $\Phi$, which concludes
the proof.
\end{proof}

\section{Applications}\label{sec:app}
\subsection{Asynchronous devices and cascade product}
We first quickly review the model of asynchronous automata due
to Zielonka \cite{Zielonka-RAIRO-TAI-87}. Besides using these automata
to accept trace languages, we also use them, as in
\cite{MukundSohoni-DC97, AdsulGSW-CONCUR20, AdsulGSW-LMCS22}, to 
locally compute relabelling
functions on input traces, similar in spirit to the sequential
letter-to-letter transducers on words. The {\em local cascade product} of asynchronous automata
from \cite{AdsulGSW-CONCUR20, AdsulGSW-LMCS22} is a natural generalization
of cascade product of sequential automata and, like in the word case 
\cite{StraubingBook}, it 
corresponds to compositions of related relabelling functions.

Recall that, in keeping with our earlier notation, $\distribution = \{\Sigma_i\}_{i \in \pset}$ is a distributed alphabet
over the set $\pset$ of processes with total alphabet $\Sigma = \cup_{i \in \pset} \Sigma_i$.

An \emph{asynchronous automaton} $\A$ over $\distribution$ (or simply over $\s$) is a tuple 
$\A = ({\{\ls_i\}}_{i \in \pset},
{\{\lt_a \}}_{a \in \alphabet}, \sinit)$ where 
\begin{itemize}
	\item for each $i \in \pset$, $\ls_i$ is a finite non-empty set of
		local $i$-states.
	\item for each $a \in \alphabet$, $\lt_a\colon \ls_a \to \ls_a$ is a deterministic joint
	transition function where
		$\ls_a = \prod_{i \in \loc(a)} \ls_i$ is the set of 
		$a$-states.
	\item with $S =  \prod_{i \in \pset} \ls_i$ as the set of all global states, $\sinit \in S$ is a designated initial global state.
\end{itemize}
For a global state $s \in S$, we write $s=(s_a, s_{-a})$ where
$s_a$ is the projection of $s$ on $\loc(a)$ and $s_{-a}$ is the projection
on the complement $\pset\setminus \loc(a)$.
For $a \in \alphabet$, the joint transition function $\lt_a\colon\ls_a \to \ls_a$
can be naturally extended to a global transition function 
$\gt_a\colon S \to S$, on global states as follows: 
$\gt_a((s_a, s_{-a}))= (\lt_a(s_a), s_{-a})$. Further we define, for
a trace $t$ over $\s$, $\gt_t\colon S \to S$ by letting $\gt_{\epsilon}$ be
the identity function on $S$, the composition $\gt_t = \gt_a \circ \gt_{t'}$ when $t=t'a$. 
The well-definedness of $\gt_t$ follows easily from the 
fact that, for a pair $(a,b)$ of independent letters, $\gt_a \circ \gt_b = 
\gt_b \circ \gt_a$.
We denote by
$\A(t)$ the global state $\gt_t(\sinit)$ reached when running $\A$ on $t$.

Asynchronous automata are used to accept trace languages. If $F$
is a subset of the set $S$ of global states, we say that 
$L(\A, F)= \{ t \mid \A(t) \in F \}$ is the \emph{trace language accepted by $\A$ 
with final states $F$}. A trace language is said to be \emph{accepted by $\A$}
if it is equal to $L(\A,F)$ for some $F \subseteq S$.

In this work, we use asynchronous automata also as 
letter-to-letter asynchronous transducers, which compute a relabelling function. Let $\g$ be a finite 
non-empty set.
The set $\s \times \g$ naturally inherits a distribution from $\distribution$:
$\loc((a,\gamma)) = \loc(a)$.  We denote by $\traces$
(resp.\ $\ltraces$) the set of traces over corresponding alphabets.
A function $\theta\colon \traces \to \ltraces$ is called a \emph{$\g$-labelling 
function} if, for every $t=(E, \leq, \lambda) \in \traces$,
$\theta(t)=(E, \leq, (\lambda, \mu)) \in \ltraces$. Thus a 
$\g$-labelling function simply decorates each event $e$ of the trace $t$
with a label $\mu(e)$ from $\g$. 

An \emph{asynchronous $\g$-transducer} over $\s$ is a tuple 
$\widehat{\A} = (\A, \{\mu_a\})$  where $\A=(\{\ls_i\}, \{\lt_a\}, \sinit)$ is
an asynchronous automaton and each $\mu_a$ ($a \in \Sigma$) is
a map $\mu_a\colon \ls_a \to \g$. The $\g$-labelling
function \emph{computed} by $\widehat{\A}$ is the following map from $\traces$ to $\ltraces$, also denoted by $\widehat{\A}$:
for $t = (E, \leq, \lambda) \in \traces$, let 
$\widehat{\A}(t)= (E, \leq, (\lambda, \mu)) \in \ltraces$ such that, 
for every $e \in E$ with
$\lambda(e)=a$, $\mu(e) = \mu_{a}(s_a)$ where $s= \A(\Da e)$.

An asynchronous automaton (resp.\ transducer) with local state 
sets $\{S_i\}$ is said to be \emph{localized at process $i$} if 
all local state sets $S_j$ with $j\neq i$ are singletons. 
It is \emph{localized} if it is localized at some process. 
Note that, in a device localized at $i$, only process $i$ carries
non-trivial information and all non-trivial transitions are on
letters in which process $i$ participates.

Now we introduce the important notion of local cascade product  of 
asynchronous transducers
\cite{DBLP:conf/fsttcs/AdsulS04,AdsulGSW-LMCS22}. 

\begin{definition}
Let $\widehat{\A}= (\{\ls_i\}, \{\lt_a\}, \sinit, \{\mu_a\})$ be
a $\g$-transducer over $\s$ and 
$\widehat{\B} = (\{\ls[Q]_i\}, \{\lt_{(a,\gamma)} \}, \sinit[q], 
\{\nu_{(a, \gamma)} \})$ be $\Pi$-transducer over $\s \times \g$. 
We define the local cascade product of $\widehat{\A}$ and $\widehat{\B}$ to
be the $(\g \times \Pi)$-transducer $\widehat{\A} \lc \widehat{\B} =
(\{\ls_i \times \ls[Q]_i\}, \{\lt[\nabla]_a \}, (\sinit,\sinit[q]), 
\{\tau_a \})$ over $\s$ where
$\nabla_a((s_a, q_a))= (\lt_a(s_a), \lt_{(a, \mu_a(s_a))}(q_a))$ and
$\tau_a\colon \ls_a \times \ls[Q]_a \to \g \times \Pi$ is defined by
$\tau_a((s_a, q_a))= (\mu_a(s_a), \nu_{(a,\mu_a(s_a))}(q_a))$.
\end{definition}

In the sequential case, that is, when $|\pset|=1$, the local cascade
product coincides with the well-known operation of cascade product of
sequential letter-to-letter transducers.

The following lemma (see \cite{AdsulGSW-LMCS22}) is easily verified.
\begin{lemma} The $(\g \times \Pi)$-labelling function 
computed by $\widehat{\A} \lc \widehat{\B}$ is the composition of
the $\g$-labelling function computed by $\widehat{\A}$ and the
$\Pi$-labelling function computed by $\widehat{\B}$: for every $t \in \traces$,
$$(\widehat{\A} \lc \widehat{\B} )(t) = \widehat{\B}( \widehat{\A}(t) )$$
\end{lemma}

\subsection{$\locpastpdl$ translation into cascade product}
Now we exploit the locality and the natural hierarchical structure
of $\locpastpdl$ formulas to translate them into local cascade 
products of localized devices.

We start with the definition of a natural relabelling function associated
with a collection $F$ of $\locpastpdl$ event formulas. Let 
$\g_F = \{\True,\False\}^F$. For each trace $t \in \traces$ and event $e$ in $t$, 
we let $\mu_F(e) \in \g_F$ be the tuple of truth-values of the formulas $\phi \in F$ at $e$, and we let $\theta^F$ be the $\g_F$-labelling function given by $\theta^F(t) = (E,\leq,(\lambda, \mu_F))$, for each
$t = (E,\leq,\lambda) \in \traces$.

We can now state an important result on event formulas in $\locpastpdl$. This result, without an analysis of global states of the resulting
construction, is already implicit in \cite{AdsulGastinSarkarWeil-CONCUR22}.

\begin{theorem}\label{thm:locpastpdl-event-formula-to-automata}
  Let $\phi$ be a $\locpastpdl$ event formula and $\theta^\phi$ be the corresponding
  $\{\True, \False\}$-labelling function.  Then $\theta^\phi$ can be computed by a local
  cascade product of localized asynchronous transducers. The number of global states in 
  this transducer is $2^{\mathcal{O}(|\varphi|)}$.
\end{theorem}

Let us highlight the main ideas required to prove \cref{thm:locpastpdl-event-formula-to-automata}.
The proof proceeds by structural induction on $\phi$, and constructs 
a cascade product $\A_\phi$ of localized transducers which computes $\theta^\phi$.
The most non-trivial case is when $\phi=\p{\pi}$ where $\pi$ is an $i$-local path
formula. Recall that $\pi$ is a regular expression involving top-level moves 
$\lmove_{i}$ and test formulas.  Let $F$ be the set of event formulas which appear in
these test formulas.
By induction, for each $\psi \in F$, $\theta^\psi$ can be computed by cascade product $\A_\psi$ of 
localized transducers. By a direct product construction, which may be seen 
as a special case of cascade product, we can construct a transducer of
the desired form, say $\A_F$, which computes $\g_F$-labelling function $\theta^F$.
We finally construct $\A_\phi$ as a cascade product of $\A_F \lc \B$ where
$\B$ is localized at $i$. In order to construct $\B$, we first convert $\pi$
into a deterministic finite-state automata $\B_i$ over alphabet $\g_{F}$.  Now we obtain
the asynchronous transducer by localizing $\B_i$ at process $i$ and suitably designing a
labelling function which computes $\theta^\phi$.  See
\cite{AdsulGastinSarkarWeil-CONCUR22} for more details.

The complexity bound is also proved by structural induction.  The sizes $|\varphi|$ and 
$|\pi|$ of event and path formulas are defined inductively as expected, with 
$|\test{\psi}|=1+|\psi|$. For a path expression $\pi$, we also define the top-level size 
$||\pi||$ which does not take into account the size of the event formulas tested in 
$\pi$: $||\test{\psi}||=1$.
The transducers for atomic formulas $\varphi=a$ with $a\in\Sigma$ have a single global
state.  The transducers for the boolean connectives $\neg,\vee,\wedge$ also have a single
global state.  Consider now the non-trivial case $\varphi=\p{\pi}$ described above.  By
induction, the number of global states of each $\A_{\psi}$ is $2^{\mathcal{O}(|\psi|)}$.
We get the local cascade product $\A_{F}$ with $2^{\mathcal{O}(||F||)}$ global states,
where $||F||=\sum_{\psi\in F} |\psi|$.  Finally, we translate the regular expression
$\pi$, considering test formulas $\test{\psi}$ for $\psi\in F$ as uninterpreted symbols,
to a DFA with $2^{\mathcal{O}(||\pi||)}$ many states. A final cascade product yields 
$\A_{\varphi}$ with $2^{\mathcal{O}(|\varphi|)}$ many global states.

Now we turn our attention to translating $\locpastpdl$ sentences into automata. A basic trace formula $\Phi$ is of the form
$\EM_{i} \phi$ where $\phi$ is an event formula. A trace $t$ satisfies
$\EM_{i} \phi$ if the last $i$-event exists and $\phi$ holds at this event.
By \cref{thm:locpastpdl-event-formula-to-automata},
we have a local cascade product $\widehat{\A}$ of localized transducers which 
computes $\theta^\phi$ and thus, records the truth value of $\phi$ at every event.
It is easy to convert $\widehat{\A}$ into an asynchronous automaton which checks 
at the final global state if the last $i$-event exists and whether $\phi$ evaluates to $\True$ at this event. 
As an arbitrary trace formula is
simply a boolean combination of formulas of the form $\EM_{i} \phi$, which
can be handled by a direct product construction, we get the following result.

\begin{theorem}\label{thm:locpastpdl-trace-formula-to-automata}
	Let $\Phi$ be a $\locpastpdl$ sentence. 
  We can construct a local cascade product of localized asynchronous automata
  which precisely accepts the language defined by $\Phi$.
  The number of global states is $2^{\mathcal{O}(|\Phi|)}$.
\end{theorem}

\subsection{Zielonka's theorem, distributed Krohn-Rhodes theorem and a 
Gossip implementation}
Now we present some important applications of the expressive completeness
of $\locpastpdl$.
The next theorem provides a new proof of the fundamental theorem of
Zielonka which characterizes regular trace languages using asynchronous
automata.

\begin{theorem}\label{thm:Zielonka-theorem}
A trace language is regular if and only if it is accepted by a local cascade 
product of localized asynchronous automata. In particular, every regular trace language can be accepted by an asynchronous automaton.
\end{theorem}
\begin{proof}
Let $L$ be a regular trace language. By \cref{thm: locpastpdl is expressively complete}, there is a $\locpastpdl$ sentence $\Phi$ which defines $L$.
By \cref{thm:locpastpdl-trace-formula-to-automata}, we can construct
a local cascade product of localized automata which accepts the language
defined by $\Phi$, that is, $L$.

Conversely, a local cascade product of localized automata is
an asynchronous automaton, which is known to 
accept a regular trace language.
\end{proof}

Another important application is a novel distributed Krohn-Rhodes theorem. 
Recall that the classical 
Krohn-Rhodes theorem \cite{KR} states that every sequential automaton 
can be simulated by a {\em cascade product} of simple automata, namely, 
two-state reset automata and permutation automata. In a two-state reset
automaton, %
the transition function of 
each letter is either the identity function or a constant function.
In contrast, in a permutation automaton, each letter induces a permutation
of the state set. 

Let $i \in \pset$. A two-state reset automaton localized at 
process $i$ is an asynchronous automaton localized at $i$ which
has two $i$-local states and the local transition on each letter
is either the identity function or a constant. Similarly, in a localized
permutation automaton, each letter of the active process induces
a permutation of its local state set while the other local state sets
are singletons.

\begin{theorem}
	Every regular trace language can be accepted by a local
	cascade product of localized two-state reset automata and
	localized permutation automata.
\end{theorem}
\begin{proof}
Note that, by essentially using the classical Krohn-Rhodes theorem,
an asynchronous automaton localized at process $i$ can be simulated
by a local cascade product of two-state reset automata localized
at $i$ and permutation automata localized at $i$. Therefore,
the theorem follows by \cref{thm:Zielonka-theorem}.
\end{proof}

Our final application concerns an implementation of a gossip
automaton/transducer.
Let $\gossip$ be the collection of constant event formulas 
$\{\Yleq{i}{j}\}_{i,j \in \pset}$. Consider the $\g_\gossip$-labelling 
function. Note that $\g_\gossip$ records, for every trace and
at each event of that trace, the truth values of the formulas in $\gossip$
in the form of an additional label from $\g_\gossip$.
So $\theta^\gossip$ keeps track of the ordering information between 
leading process events.
By a gossip transducer, we mean any asynchronous transducer
which computes $\theta^\gossip$
and hence keeps track of 
the latest gossip/information \cite{MukundSohoni-DC97} in a 
distributed environment. 

\begin{theorem}
The $\g_\gossip$-labelling function $\theta^\gossip$ can be computed by an asynchronous
transducer which is a local cascade product of 
localized asynchronous transducers.
\end{theorem}
\begin{proof}
By \cref{Yconstants theorem}, event formulas in $\gossip$ can be expressed in
$\locpastpdl$. Further, by \cref{thm:locpastpdl-event-formula-to-automata},
for each of these $\locpastpdl$ formulas $\phi$, the corresponding
$\theta^\phi$ can be computed by a local cascade product of 
localized asynchronous transducers. Combining these together through
a direct product, we obtain an asynchronous transducer
of the desired form which computes $\theta^\gossip$.
\end{proof}

The existing constructions \cite{Zielonka-RAIRO-TAI-87,MukundSohoni-DC97}  
of a gossip transducer are intrinsically non-local and non-hierarchical.
It is important to note that our construction is quite 
inefficient in terms of the size of the resulting transducer. The key point
of our construction is to show the existence of a gossip implementation using
compositions of labelling functions each of which is localized at 
some process. It was not at all clear that such an implementation
could exist!

\section{Conclusion}\label{sec:conclusion}
Our main result is the identification of a local, past-oriented fragment 
of propositional dynamic logic for traces called $\locpastpdl$, which 
is expressively complete with respect to regular trace languages. 
The natural hierarchical structure of $\locpastpdl$ formulas 
allows a modular translation into a 
cascade product of localized automata. 
We have also given many important applications such as
a new proof of Zielonka's theorem, a novel distributed version of 
Krohn-Rhodes theorem and a hierarchical implementation of a 
gossip automaton. 

A natural question is to identify a fragment of $\locpastpdl$
which matches the first-order logic in expressive power. 
Another promising future direction is to characterize the precise power of
local cascade products of localized two-state reset automata. 
Extending these results with %
future-oriented path 
formulas, and to the setting of infinite traces are also 
interesting directions.

It remains to investigate the utility of $\locpastpdl$ specifications 
from a more practical and empirical viewpoint. %
Our modular and efficient translation of these specifications should provide
significant improvements in applications such as synthesizing
distributed monitors and verifying properties of concurrent programs.

\bibliographystyle{plainurl}

\end{document}